\documentclass[11pt,draftcls,onecolumn]{IEEEtran}
\usepackage{array}
\usepackage{cite}
\usepackage{calc}
\usepackage{graphicx}
\usepackage{psfrag}
\usepackage[tight,footnotesize]{subfigure}
\usepackage{stfloats}
\usepackage{url}
\usepackage{subfig}
\usepackage{longtable}
\usepackage{amsmath,amsfonts,amssymb,amsbsy,bm,paralist,theorem,ifthen,color}
\usepackage{algorithmic,algorithm}
\usepackage{psfrag,cite,amsmath,graphicx,epsfig,latexsym,subfigure,color,epstopdf}
\usepackage{booktabs}
\usepackage{multirow}
\usepackage{verbatim}
\usepackage{float,nicefrac}

\newcommand{\beq}{\begin{equation}}
\newcommand{\eeq}{\end{equation}}

\newtheorem{theorem}{Theorem}
\newtheorem{lemma}{Lemma}
\newtheorem{corollary}{Corollary}

\newtheorem{definition}{Definition}

\theoremstyle{remark}
\newtheorem{remark}{Remark}

\title{Learning to Optimize:  Training Deep Neural Networks for Wireless Resource Management}
\author{\authorblockN{Haoran Sun$^{*}$, Xiangyi Chen$^{*}$, Qingjiang Shi, Mingyi Hong, Xiao Fu, and Nicholas D. Sidiropoulos}
\thanks{H. Sun, X. Chen and M. Hong are with the Department of  Electrical and Computer Engineering, University of Minnesota, Minneapolis, MN 55455, USA. Emails: \{sun00111, chen5719, mhong\}@umn.edu}
\thanks{Q. Shi is with the College of Electronic and Information Engineering, Nanjing University of Aeronautics and Astronautics, Nanjing, China. Email: qing.j.shi@gmail.com}
\thanks{X. Fu is with the School of Electrical Engineering and Computer Science, Oregon State University, Corvallis, OR 97331, USA. Email: xiao.fu@oregonstate.edu}
\thanks{N. D. Sidiropoulos is with the Department of Electrical and Computer Engineering, University of Virginia, Charlottesville, VA 22904, USA. Email: nikos@virginia.edu}
\thanks{$^{*}$ The first two authors contributed equally to this work. H. Sun, X. Chen and M. Hong are supported by NSF grants CMMI-1727757, CCF-1526078, and an AFOSR grant 15RT0767.}
\thanks{An early version of this  paper has been  published in IEEE 18th International Workshop on Signal Processing Advances in Wireless Communications (SPAWC 2017) { \cite{sun2017learning}}}}

\begin{document}
\maketitle
\begin{abstract}
For the past couple of decades, numerical optimization has played a central role in addressing wireless resource management problems such as power control and beamformer design. However, optimization algorithms often entail considerable complexity, which creates a serious gap between theoretical design/analysis and real-time processing. To address this challenge, we propose a new learning-based approach. The key idea is to treat the input and output of a resource allocation algorithm as an unknown non-linear mapping and use a deep neural network (DNN) to approximate it. If the non-linear mapping can be learned accurately by a DNN of moderate size, then resource allocation can be done in almost \emph{real time} -- since passing the input through a DNN only requires a small number of simple operations. 

In this work, we address both the thereotical and practical aspects of DNN-based algorithm approximation with applications to wireless resource management.
We first pin down a class of optimization algorithms that are `learnable' in theory by a fully connected DNN. Then, we focus on DNN-based approximation to a popular power allocation algorithm named WMMSE (Shi {\it et al} 2011). We show that using a DNN to approximate WMMSE can be fairly accurate -- the approximation error $\epsilon$ depends mildly [in the order of $\log(1/\epsilon)$] on the numbers of neurons and layers of the DNN. On the implementation side, we use extensive numerical simulations to demonstrate that DNNs can achieve orders of magnitude speedup in computational time compared to state-of-the-art power allocation algorithms based on optimization.
\end{abstract}

\section{Introduction}
Resource management tasks, such as transmit power control, transmit/receive beamformer design, and user admission control, {are critical} for future wireless networks. Extensive research has been done to develop various resource management schemes; see the recent overview articles \cite{hong12survey,bjornson13}. For decades, {numerical} optimization has played a central role in addressing wireless resource management problems. Well-known optimization-based algorithms for such purposes include those developed for power control (e.g., iterative water-filling type algorithms \cite{yu02a,scutari08a}, interference pricing \cite{Schmidt09}, SCALE \cite{papand09}), transmit/receive beamformer design (e.g., the WMMSE algorithm \cite{shi11WMMSE_TSP}, pricing-based schemes \cite{kim11}, semidefinite relaxation based schemes \cite{luo10SDPMagazine}), admission control (e.g., the deflation based schemes \cite{liu13deflation}, convex approximation based schemes \cite{Matskani08}), user/base station (BS) clustering (e.g., the sparse optimization based schemes \cite{hong12sparse}), just to name a few. Simply speaking, these algorithms all belong to the class of {\it iterative algorithms}, which take a given set of real-time network parameters like channel realizations and signal to noise ratio specifications as their inputs, run a {number of} (potentially costly) iterations, and produce the ``optimized'' resource allocation {strategies} as their outputs.
\begin{figure}[htb]
	\begin{minipage}[b]{.45\linewidth}
	 
		\centering
		\centerline{\includegraphics[width=\linewidth]{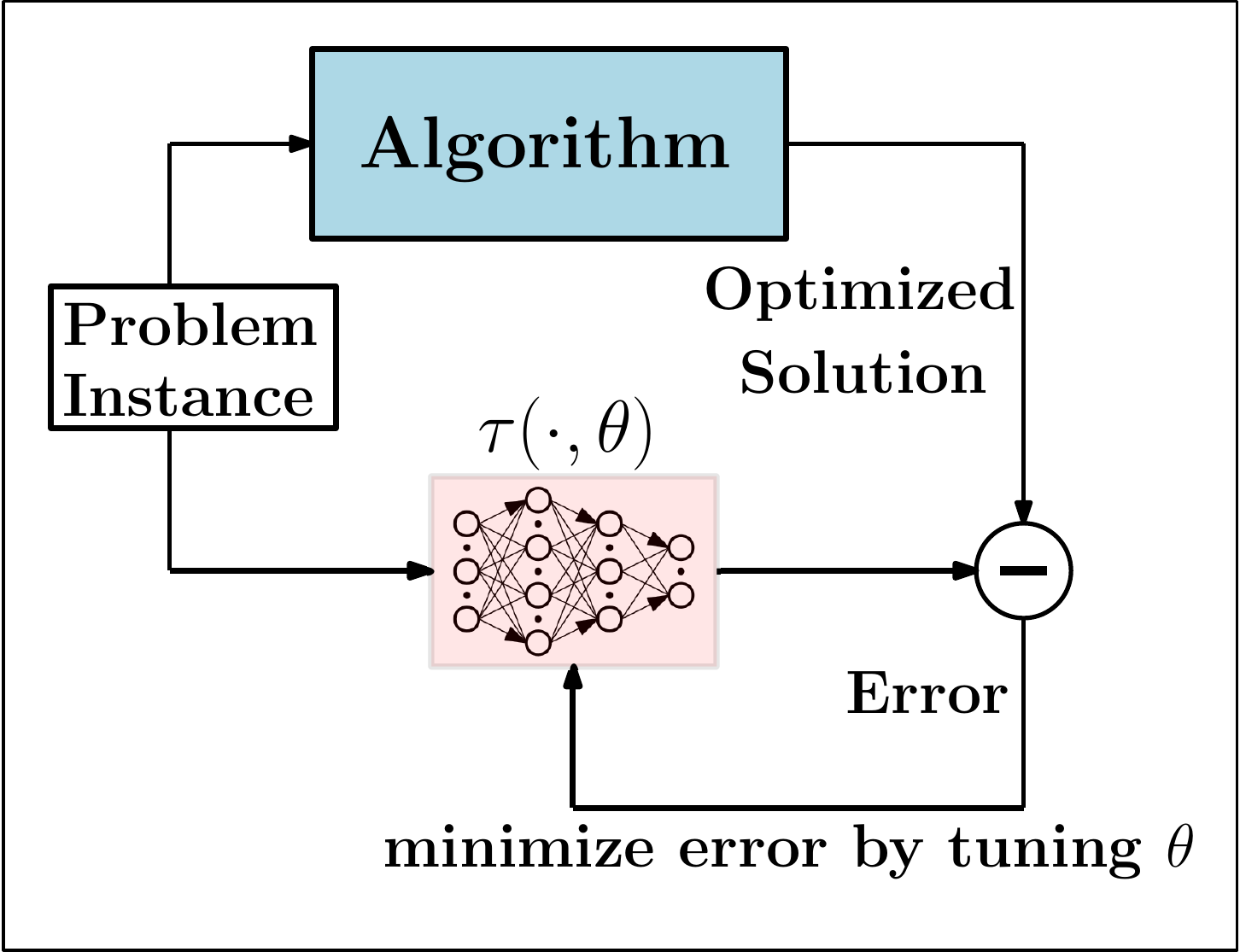}}
		\centerline{\footnotesize  (a) Training Stage} \medskip
	\end{minipage}
	\hfill
	\begin{minipage}[b]{0.5\linewidth}
		\centering
		\centerline{\includegraphics[width=.9\linewidth]{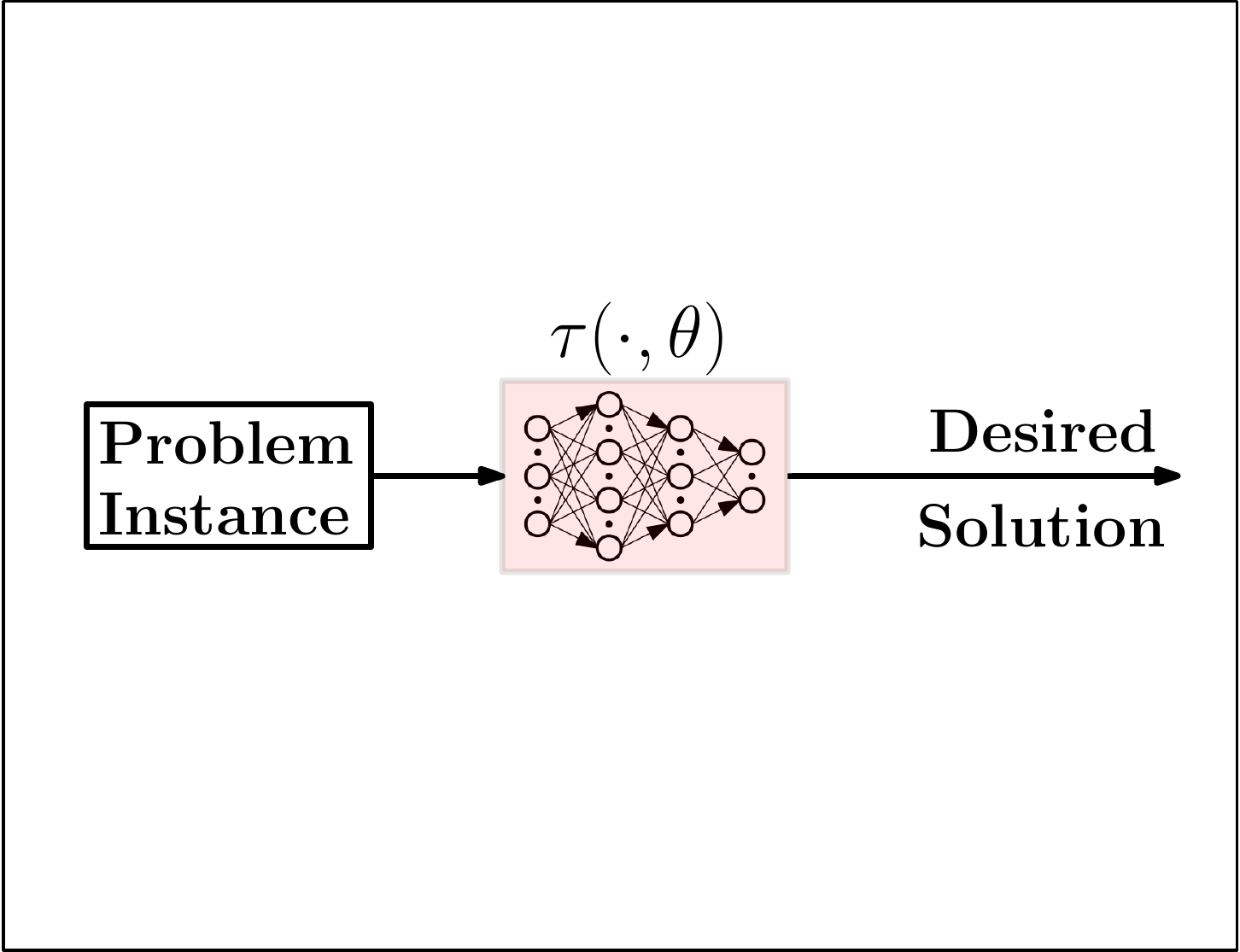}}
		\centerline{ \footnotesize (b) Testing Stage}\medskip
	\end{minipage}	
 
	\caption{\footnotesize The Proposed Method. In the figure $\tau(\cdot,\theta)$ represent a DNN parameterized by $\theta$.}
 
	\label{fig:TrainDNN}
\end{figure}

Despite the fact that excellent performance of many of these algorithms has been observed through numerical simulations and theoretical analysis, 
{implementing them in real systems still faces many serious obstacles. In particular, the high computational cost incurred by these algorithms has been one of the most challenging issues.  For example, WMMSE-type algorithms require complex operations such as matrix inversion and bisection in each iteration \cite{shi11WMMSE_TSP, baligh2014cross}. Water-filling type algorithms such as \cite{yu02b} and interference pricing algorithms \cite{Schmidt09} involve singular value decomposition at each iteration (when dealing with Multiple-Input-Multiple-Output (MIMO) interference systems). Algorithms such as the deflation-based joint power and admission control require successively solving a series of (possibly large-scale) linear programs.  
The computationally demanding nature of these algorithms makes real-time implementation challenging, because wireless resource management tasks are typically executed in a time frame of milliseconds (due to the fast changing system parameters  such as channel conditions, number of users, etc.).

In this work, we propose a new machine learning-based approach for wireless resource management; see Figure \ref{fig:TrainDNN} for an illustration of the system schematics. The main idea is to treat a given resource optimization algorithm as a ``black box", and try to {\it learn} its input/output relation by using a {\emph{deep neural network} (DNN)} \cite{lecun2015deep}. 
{Our motivation is that if a network with several layers can well approximate a given resource management algorithm,
then passing the algorithm input through the network to produce an output can be rather economical in computation -- it only requires several layers of simple operations such as matrix-vector multiplications. Therefore, such an approximation, if accurate, can reduce processing time for wireless resource allocation substantially.
In addition, training such a network is fairly convenient, since training samples are easily obtainable by running the optimization algorithm of interest on simulated data.

In the machine learning community, there have been several attempts of approximating an iterative optimization algorithm using DNNs in recent years. The work in \cite{gregor2010learning} proposed to use a multilayer network to approximate the iterative soft-thresholding algorithm (ISTA) for {sparse optimization} \cite{Beck:2009:FIS:1658360.1658364}. 
In particular, the ISTA algorithm is first ``unfolded" iteration by iteration as a chain of operations, and each of its iterations is mimicked by a layer of the network. Then the parameters of the network are learned by offline training.  The authors of \cite{hershey2014deep,sprechmann2013supervised} applied similar ideas for nonnegative matrix factorization and the alternating direction method of multipliers (ADMM), respectively. Recently, the authors of \cite{andrychowicz2016learning} proposed to learn the per-iteration behavior of gradient based algorithms, and the authors of \cite{li2016learning} proposed to use reinforcement learning for the same task. In a nutshell, these algorithms are more or less based on the same ``unfolding'' idea, which is feasible either because the iterations have simple structure such as in ISTA or due to the fact that gradient information is assumed to be known like in \cite{andrychowicz2016learning} and \cite{li2016learning}. However, for complex algorithms, it is no longer clear whether an iteration can be explicitly modeled using a structured layer of a neural network. Further, it is important to note that none of these existing works have provided rigorous theory to back their approaches -- e.g., it is unclear whether or not algorithms such as ISTA can be accurately approximated using the unfolding idea.

{It is worth noting that there have been a number of works that apply DNN for various communication tasks. 
For example, the recent works in \cite{o2016recurrent,farsad2017detection,west2017deep,o2017deep} have demonstrated promising performance of using DNN in a number of tasks such as anomaly detection and decoding. However, the focus in these papers is signal modeling rather than algorithm approximation.
The very recent work in \cite{samuel2017deep} proposes a DNN architecture to approximate an algorithm for the MIMO detection problem, but the idea, again, follows the unfolding work in \cite{gregor2010learning}.

Our approach is very different from the unfolding idea: We propose to employ a generic fully connected DNN to approximate a wireless resource allocation algorithm, rather than unfolding the algorithms of interest. Unlike the algorithms with relatively simple structure that could be easily unfolded, wireless resource management algorithms often entail computationally heavy iterations involving operations such as matrix inversion, SVD, and/or bi-section. 
Therefore, their iterations are not amenable to approximation by a {\it single} layer of the network. 

Overall, the DNN-based algorithm approximation is an emerging and largely unexplored topic. Therefore its success will be critically dependent on the answer to the following challenging research questions: What kind of optimization algorithms are ``learnable", meaning they can be approximated arbitrarily well by a neural network? How many layers / neurons are needed to approximate a resource allocation algorithm accurately? How robust or sensitive is the method in practice?}

\smallskip

\noindent{\bf Contributions.}  
The main contribution of this work is {three-fold: First, we propose the first deep learning based scheme for real-time resource management over interference-limited wireless networks, which bridges the seemingly unrelated areas of machine learning and wireless resource allocation (in particular, power control over interference networks). 
Second, unlike all existing works on approximating optimization algorithms such as those using unfolding, our approach is justified by rigorous theoretical analysis.
We show that there are conditions under which an algorithm is leanrable by a DNN, and that the approximation error rate is closely related to the `depth' and `width' of the employed DNN.
These theoretical results provide fundamental understanding to the proposed ``learn to optimize" approach. They also offer insights and design guidelines for many other tasks beyond wireless resource allocation. Third, extensive numerical experiments using simulated data and {\it real DSL data} are conducted to demonstrate the achievable performance of the proposed approach under a variety of scenarios. As a proof-of-concept, the results provided in this paper indicate that DNNs have great potential in the real-time management of the wireless resource.

To promote reproducible research, the codes for generating most of the results in the paper are made available on the authors' website: \url{https://github.com/Haoran-S/TSP-DNN}. 

\smallskip

\noindent{\bf Notation.} We use $I_{X}(x)$ to denote the indicator function for a set $X$; the function takes value $1$ if $x\in X$, and its takes value $0$ if $x\notin X$; We use $\mathcal{N}(m, \sigma^2)$ to denote normal distribution with mean $m$ and variance $\sigma^2$. {A Rectified Linear Unite (ReLU)  is defined as $\mbox{ReLU}(x) =\max\{0, x\}$; A binary unit, denoted as 
	$\mbox{Binary}(x)$, returns $1$ if $x$ evaluates as true, and returns $0$ otherwise.} We use $\mathbb{R}_{+}$ to denote the set of all nonnegative numbers. 

\section{Preliminaries}
\subsection{System Model} 
We consider the following basic interference channel (IC) power control problem,  for a wireless network consisting of $K$ single-antenna transceivers pairs.  Let $h_{kk}\in\mathbb{C}$ denote the direct channel between transmitter $k$ and receiver $k$, and $h_{kj}\in\mathbb{C}$ denote the interference channel from  transmitter~$j$ to receiver $k$. All channels are assumed to be constant in each resource allocation slot. Furthermore, we assume that the transmitted symbol of transmitter $k$ is a Gaussian random variable with zero mean and variance $p_k$ (which is also referred to as the transmission power of transmitter $k$). Further, suppose that the symbols from different transmitters are independent of each other. Then the signal to interference-plus-noise ratio (SINR) for each receiver $k$ is given by
$$\mbox{sinr}_k\triangleq \frac{|h_{kk}|^2p_k}{\sum_{j\neq k}|h_{kj}|^2p_j+\sigma_k^2},$$ 
where $\sigma_k^2$ denotes the noise power at receiver $k$.

We are interested in power allocation for each transmitter so that the weighted system throughput is maximized. Mathematically, the problem can be formulated as the following {\it nonconvex} problem  
\begin{equation}\label{eq:sum_rate_IA}
\begin{split}
\max_{p_1,\ldots, p_K}\quad &\sum_{k=1}^K \alpha_k\log\left(1+\frac{|h_{kk}|^2p_k}{\sum_{j\neq k}|h_{kj}|^2p_j+\sigma_k^2}\right)\\
\textrm{s.t.} \quad &0\leq p_k\leq P_{\max},~ \forall \; k=1,2,\ldots, K,
\end{split}
\end{equation}
where $P_{\max}$ denotes the power budget of each transmitter; $\{\alpha_k>0\}$ are the weights. Problem \eqref{eq:sum_rate_IA} is known to be NP-hard \cite{luo08a}. Various power control algorithms have been proposed \cite{shi11WMMSE_TSP,Schmidt09,papand09}, among which the WMMSE algorithm has been very popular. In what follows, we introduce a particular version of the WMMSE applied to solve the power control problem \eqref{eq:sum_rate_IA}.

\subsection{The WMMSE Algorithm} 
The WMMSE algorithm converts the weighted sum-rate (WSR) maximization problem to a higher dimensional space where it is easily solvable, using the well-known MMSE-SINR equality\cite{verdu98}, i.e., $\mbox{mmse}_k=\nicefrac{1}{(1+\mbox{sinr}_k)}$, where ``$\mbox{mmse}_k$" denotes the minimum-mean-squared-error (MMSE) of user $k$. 

Note that the original WMMSE \cite{shi11WMMSE_TSP}  is designed for problems whose design variables are beamformer vectors with complex entries. In this paper, in order to simplify our subsequent implementation of DNN based schemes, we modify the  algorithm so that it can work in the real domain. Specifically, we first observe that the rate function remains unchanged if $h_{kj}$ is replaced by $|h_{kj}|$.  Using this observation, one can show (by following the steps in \cite[Theorem 1]{shi11WMMSE_TSP}) that problem \eqref{eq:sum_rate_IA} is equivalent to the following weighted MSE minimization problem
{\setlength\abovedisplayskip{3pt}
	\setlength\belowdisplayskip{3pt}
 
\begin{equation}\label{eq:virtual_IA_1}
\begin{split}
\min_{\{w_k,u_k,v_k\}_{k=1}^K}\quad &\sum_{k=1}^K \alpha_k\left(w_k e_k-\log(w_k)\right)\\
\textrm{s.t.} \quad &0\leq v_k\leq \sqrt{P_k}, ~k=1,2,\ldots, K.
\end{split}
\end{equation}}
where the optimization variables are all real numbers, and we have defined 
\begin{equation}\label{eq:mse}
e_k =(u_k|h_{kk}|v_k-1)^2+\sum_{j\neq k} (u_k|h_{kj}|v_j)^2+\sigma_k^2u_k^2
\end{equation}
Here by {\it equivalent} we meant that all stationary solutions of these two problems are identical to each other. 

The WMMSE solves \eqref{eq:virtual_IA_1} using the block coordinate descent method \cite{bertsekas97}, i.e., each time optimizing one set of variables while keeping the rest fixed; see Fig. \ref{fig:pseudo_code1} for its detailed steps. It has been shown in \cite{shi11WMMSE_TSP} that the WMMSE is capable of reaching a stationary solution of problem \eqref{eq:sum_rate_IA}. 
Note that since the interfering multiple-access channel (IMAC) can be viewed as a special IC with co-located receivers, the WMMSE algorithm in Fig. 1 can be applied to solving the power allocation problem of IMAC as well \footnote{In the IMAC, we assume that within each cell the BS does not perform successive interference cancellation.}.

\begin{figure}[htbp]
	\centering
	\begin{tabular}{|p{5.2in}|}
		\hline
		~\,1.\quad Initialize $v_k^0$ such that $0\leq v_k^0\leq \sqrt{P_k}$, $\forall~k$;\\
		~\,2.\quad Compute $u_k^0= \frac{|h_{kk}|v_k^0}{\sum_{j=1}^K |h_{kj}|^2(v_j^{0})^2+\sigma_k^2}$, $\forall~k$;\\
		~\,3.\quad Compute $w_k^0=\frac{1}{1-u_k^0|h_{kk}|v_k^0}$, $\forall~k$;\\
		~\,4.\quad Set $t=0$\\
		~\,5.\quad \textbf{repeat}\\
		~\,6.\quad\quad $t=t+1$~~~//iterations\\
		~\,7.\quad\quad Update $v_k$: \\ \quad \quad \quad $v_k^{t}=\left[\frac{\alpha_kw_k^{t-1}u_k^{t-1}|h_{kk}|}{\sum_{j=1}^K \alpha_jw_j^{t-1} (u_j^{t-1})^2|h_{jk}|^2}\right]_{0}^{\sqrt{P_{\max}}}, \; \forall~k$;\\
		~\,8.\quad\quad Update $u_k$: $u_k^t = \frac{|h_{kk}|v_k^t}{\sum_{j=1}^K |h_{kj}|^2(v_j^{t})^2+\sigma_k^2}, \; \forall~k$;\\
		~\,9.\quad\quad Update $w_k$: $w_k^t=\frac{1}{1-u_k^t|h_{kk}|v_k^t}, \; \forall~k$;\\
		10.\quad\textbf{until} Some stopping criteria is met \\
		11.\quad\textbf{output} $p_k = (v_k)^2, \; \forall~k$; \\
		\hline
	\end{tabular}
	\caption{Pseudo code of WMMSE for the scalar IC.} \label{fig:pseudo_code1}
\end{figure}

\section{The Proposed Approach}
In this section, we present the proposed approach that uses DNN to approximate WMMSE. We first present two theoretical results to justify our approach, followed by a detailed description of how the DNN is set up, and how training, validation and testing stages are carried out.

\subsection{Universal Approximation}
At this point, it remains unclear whether a  multi-layer neural network can be used to approximate the behavior of a given iterative algorithm, like WMMSE,  for solving the nonconvex optimization problem \eqref{eq:sum_rate_IA}. The answer to such a question turns out to be nontrivial. 
To see where the challenge is, let us consider the following example. 

\noindent{\bf Example.} We consider approximating the behavior of the classic gradient descent (GD)  for solving the following problem 
\begin{align}\label{eq:example}
\min_{x\in X} f(x)=(x^2-z)^2
\end{align}
where $x$ is the optimization variable and $z$ is the problem parameter.  
The GD  iteration  is given by the following 
$$x^{t+1}=x^t-\alpha \nabla f(x^t) = x^t - 2\alpha  x^t((x^t)^2-z), \quad t=0, \cdots, T$$ 
where $t$ is the iteration number; $x^0$ is the randomly generated initial solution; $\alpha>0$ is the stepsize; $T$ is the total number of iteration. 

Given a set of training data points $\{z^{(i)}, (x^T)^{(i)}\}$, we use a simple three layer neural network to approximate the relationship $z\to x^T$, which characterizes the behavior of  GD. Unfortunately, the network learns a model that  only outputs zero, regardless what the inputs are; see Figure \ref{fig:network} (a). 
\begin{figure}[htb]
	\begin{minipage}[b]{.45\linewidth}
		\centering
		\centerline{\includegraphics[width=\linewidth]{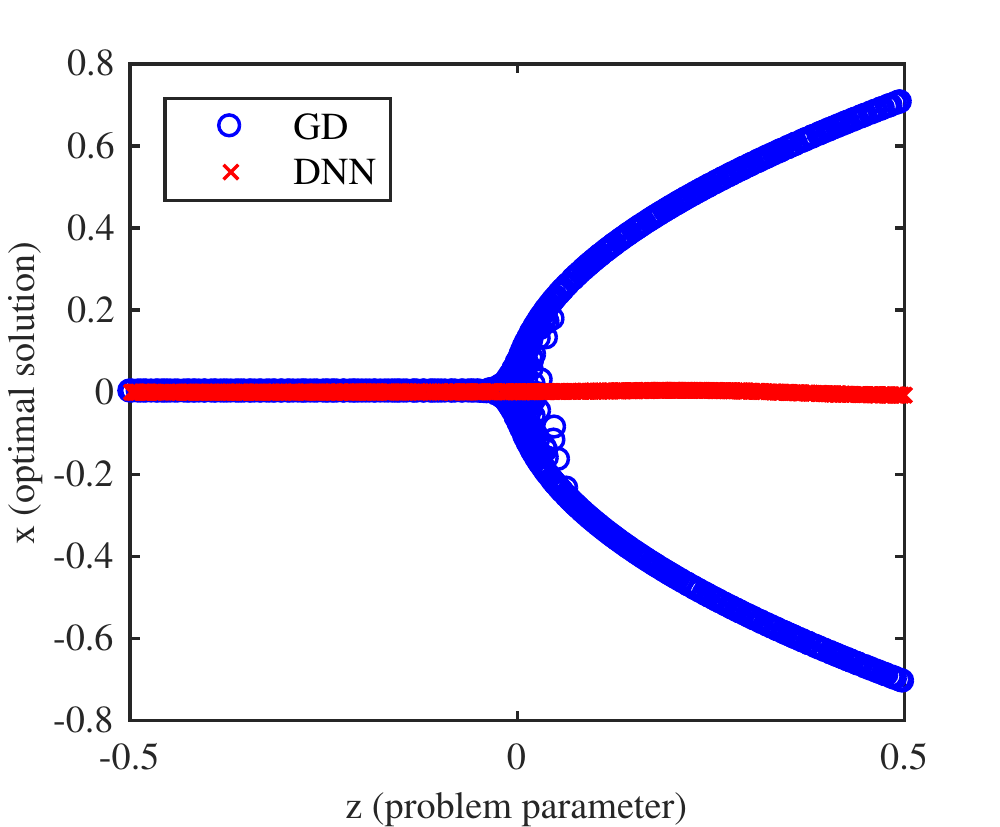}}
	\end{minipage}
	\hfill
	\begin{minipage}[b]{0.45\linewidth}
		\centering
		\centerline{\includegraphics[width=\linewidth]{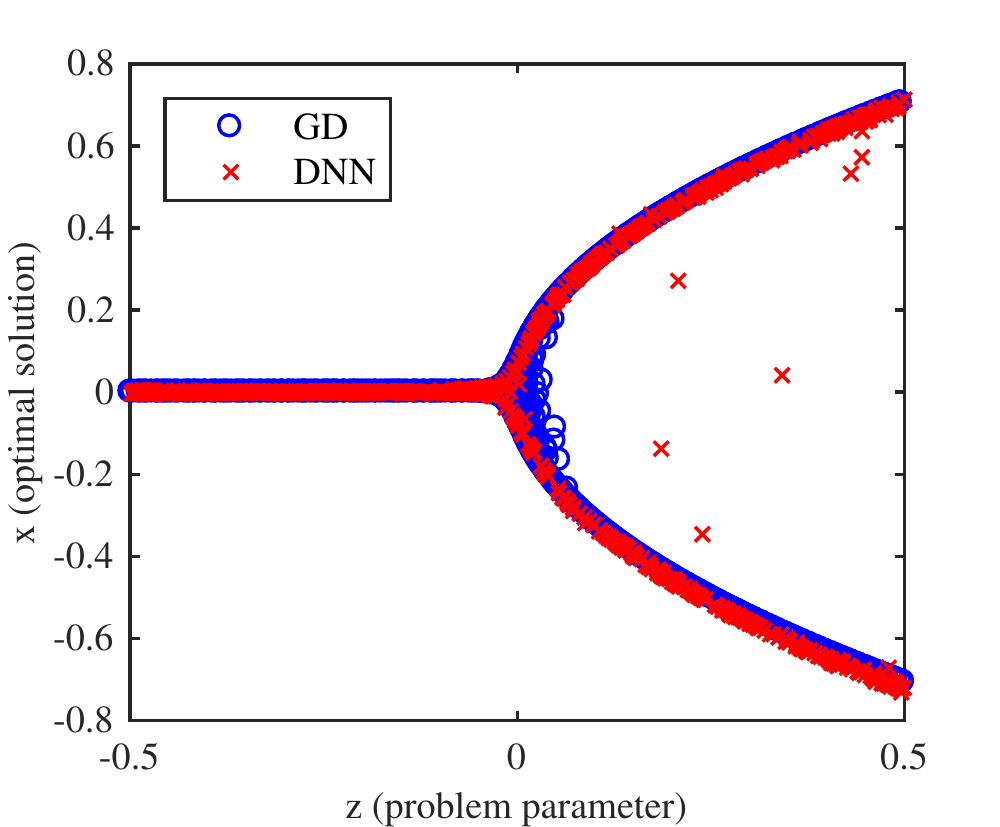}}
	\end{minipage}
 
	\caption{\footnotesize  Illustration of different models learned from the GD iteration for solving problem \eqref{eq:example}. The blue curves in the figures represent the input output relationship for GD iterations (i.e., the training samples); the red curves represent that for the learned neural networks. We have chosen $T=3000$  when generating all training samples. Left: The model learned by approximating the relationship $z\to x^T$; Right: The model learned by approximating the relationship $(x^0,z)\to x^T$}
 
	\label{fig:network}
\end{figure}

The reason that we are not learning the correct behavior of GD is that, problem \eqref{eq:example} is a nonconvex problem with {\it multiple isolated} stationary solutions $\pm\sqrt{z}$ when $z>0$.  Depending on what the initial solution $x^0$ is, the GD can converge to either a positive solution or the negative one. This behavior confuses the neural network, so that only the middle point is learned, which turns out to be a non-interesting solution. 

To resolve such an issue, we consider learning a new relationship $(x^0,z)\to x^T$ by using the same three layer network. Interestingly, now the behavior of GD can be accurately approximated, and the results are shown in Figure \ref{fig:network} (b). \hfill $\blacksquare$

The above example indicates that a rigorous theoretical analysis is needed to characterize how we should use neural networks to learn optimization algorithms.  
To begin our analysis, let  $NET_N(x^0,z)$ denote the output of a feedforward network, which has one hidden layer with $N$ sigmoid activation functions, and has $(x^0, z)$ as input. Then consider a {\it finite step} iterative algorithm, in which the relationship between $z^t$ and $z^{t+1}$ (i.e., its $t^{\rm th}$ and ${t+1}^{\rm th}$ iterates) is given by
\begin{equation}
x^{t+1} = f^t(x^t, z) , \; t=0,1,2,\ldots,T
\end{equation}
where $f^t$ is a continuous mapping representing the $t^{\rm th}$ iteration;  $z \in Z$ is the problem parameter; $x^t, x^{t+1}\in X$ where $Z$ and $X$ represent the parameter space and the feasible region of the problem, respectively. 
We have the following result as an extension of the well-known universal approximation theorem for multilayer feedforward networks  \cite{hornik1989multilayer}.
\begin{definition}
	Denote $i^{th}$ coordinate of $y$ as $y_i$, a function $y=f(x), x \in R^m, y \in R^n$ is a continuous mapping when $y_i = f_i(x)$ is a continuous function for all $i$.
\end{definition}
 
\begin{theorem}\label{thm:app} (Universal approximation theorem for iterative algorithm)
	Suppose that $Z$ and  $X$ are certain compact sets. Then the mapping from
	problem parameter $z$ and initialization $x^0$ to final output $x^T$, i.e.,
	\begin{equation*} 
	x^T = f^T(f^{T-1}(\ldots f^1(f^0(x^0,z),z)\ldots,z),z) \triangleq F^T(x^0,z)
	\end{equation*}
	can be accurately approximated by $NET_N(x^0,z)$ (which represents an N hidden units network with sigmoid activation function) in the following sense: For any given error $\epsilon>0$, there exists a positive constant $N$ large enough such that 
	
	\begin{equation} \label{eq:error}
	\sup_{(x^0,z) \in X \times Z}\|NET_N(x^0,z) - F^T(x^0,z)\|\le \epsilon.
	\end{equation}
\end{theorem}

\noindent{\bf Proof.} The proof is relatively straightforward. Since every $f^t$ is a continuous mapping and composition of continuous mappings is still a continuous mapping, we have for any $i$, the $i^{\rm th}$ coordinate of $x^T$ is a continuous function of $(x^0,z)$, denote by $x^T_i = F^T_i(x^0,z)$. From \cite[Theorem 2.1]{hornik1989multilayer}, for any compact set $X \times Z$ and any $\delta_i > 0$, there exist $N_i$ such that

\begin{equation}
	\sup_{(x_0,z) \in X \times Z} |NET_{N_i}(x_0,z) - F^T_i(x_0,z)|\le \delta_i.
\end{equation}

Suppose $x^T$ has dimension $L$. Then by stacking $L$ of the above networks together, we can construct a larger network $\hat{x} = NET_N(x_0,z)$ with $N = \sum_{i}^L N_i$. Given $\epsilon$, we can choose $\delta_i=\frac{\epsilon}{\sqrt{L}}$, then we obtain (\ref{eq:error}). \hfill
$\blacksquare$

\begin{remark}
	Our result reveals that for any deterministic algorithm whose iterations represent continuous mappings, we can  include its initialization as an {\it additional input feature}, and learn the algorithm behavior by a well trained neural network. 
	If the underlying optimization problem is nonconvex [as the WSR maximization problem \eqref{eq:sum_rate_IA}], having the  initialization as an input feature is necessary, since without it the mapping $(z)\to x^T$ is not well defined (due to the possibility of the algorithm converging to multiple isolated solutions). Alternatively, it is also possible to learn the mapping $z\to x^T$, for a {\it fixed} initialization $x^0$. 

\end{remark}

\begin{remark}
{It can be verified that each iteration of  WMMSE represents a continuous mapping, and that the optimization variable lies in a compact set. Therefore, by assuming that the set of channel realizations $\{h_{ij}\}$ lies in a compact set,  using Theorem \ref{thm:app} we can conclude that WMMSE can be approximated arbitrarily well by a  feedforward network with a single hidden layer.  }  
\hfill $\blacksquare$
\end{remark}

\subsection{Approximation Rate of WMMSE via DNN}

Our theoretical result in the previous section indicates that a large class of algorithms, including WMMSE, can be approximated  very well by a neural network with three layers. However, such a network may not be implementable because it is not clear whether a {\it finite-size} network is able to achieve a high quality approximation. Therefore,  a more relevant and challenging question is: can we rigorously characterize the network size (i.e., the number of layers and total number of neurons) required to achieve a given approximation error?

In this section, we show that the answer is affirmative.
We will use the WMMSE algorithm as an example, and provide sharper approximation guarantees than what has been obtained in Theorem 1. In particular, we show that under certain network structure, the size of the network scales with the approximation error $\epsilon$ in the rate of $\mathcal{O}(\log(1/\epsilon))$.

The main steps of the analysis is as follows: {\bf S1)} Construct simple neural networks that consist of ReLUs and Binary Units  to approximate multiplication and division operations; {\bf S2)} Compose these small neural networks to approximate a rational function representing one iteration of the algorithm; {\bf S3)} Concatenate these rational functions to approximate the entire algorithm; {\bf S4)} Bounding the error propagated from the first iteration to the last one.  
All the proofs in this section are delegated to the Appendix.

Our first two lemmas construct neural networks that can approximate multiplication and division operations. 
\begin{lemma} \label{lemmama}
	For any two positive numbers $X_{\max}, Y_{\max}$, define
	\begin{align*}
	S &:= \{(x,y)\mid  0 \leq x \leq X_{\max}, 0 \leq y \leq Y_{\max} \}\\
	m &:= \lceil \log(X_{\max}) \rceil.
	\end{align*}
	There exists a multi-layer neural network with input $(x,y)$ and output $NET(x,y)$ satisfying the following relation 
	\begin{equation}
	\max_{(x,y) \in S}|xy - NET(x,y)| \leq \frac{Y_{max}}{2^n},
	\end{equation}
	where the network has $O(m+n)$ layers and $O(m+n)$ binary units and ReLUs.
\end{lemma}
\begin{lemma} \label{lemmada}
	For any positive number $Y_{\max}$ $Z_{\max}$, define  
	\begin{align*}
	S & := \left\{(x,y)\mid {x}/{y} \leq Z_{\max}, x \geq 0, \; 0 < y \leq Y_{max}\right\}\\
	m & := \lceil \log(Z_{\max}) \rceil.
	\end{align*}
	There exists a neural network with input $(x,y)$ and ouput $ NET(x,y)$  satisfying the following relation 
	\begin{equation}
	\max_{(x,y) \in S}\;  \left|\frac{x}{y} - NET(x,y) \right| \leq \frac{1}{2^n},
	\end{equation}
	where the network has $O(m+n)$ layers and $O(m+n)$ binary units and ReLUs.
\end{lemma}

We note that in \cite{whydnn} the authors have shown that the function $z = x^2$ can be approximated by a neural network with ReLU units and binary units. 
Our construction of $xy$ follows that in \cite{whydnn}, while the construction of the division operation $x/y$ appears to be new.

At this point, it appears that we can concatenate a series of `division' and `multiplication' networks to approximate any rational function. 
However, care has to be exercised here since when performing composition, the approximation errors from the inner functions will {\it propagate} to the outer ones. It is essential to control these errors in order to obtain the overall approximation bounds. Below we provide analysis to characterize  such an error propagation process.
\begin{lemma} \label{lemmaed}
	Given positive numbers $Y_{\min}$, $Z_{\max}$, $\epsilon_1$ and $\epsilon_2$. Suppose that $(x,y)$ satisfies the following relations
	$$0 < Y_{\min} < y - \epsilon_2, \; 0 \leq  \epsilon_1 \leq x, \; 0 \leq \epsilon_2 \leq y , \; {x}/{y} \leq Z_{\max}.$$
	Then we have 
	\begin{align} \label{diverrorprop}
	&\left|\frac{x}{y} - \frac{x \pm \epsilon_1}{y \pm \epsilon_2}\right|  \leq \frac{Z_{\max} + 1}{Y_{\min}} \max(\epsilon_1,\epsilon_2).
	\end{align}
	
\end{lemma}

\begin{lemma} \label{lemmaem}
	Suppose positive numbers $X_{\max}$, $Y_{\max}$, $\epsilon_1$ and $\epsilon_2$ satisfy the following relations 
	\begin{align}
		&0 \leq x \leq X_{\max},\quad 0 \leq y \leq Y_{\max},\\
		& \max(\epsilon_1,\epsilon_2) \leq \max(X_{\max},Y_{\max}),  \; \epsilon_1,\epsilon_2 \geq 0.
	\end{align}
	Then we have
	\begin{align} \label{mulerrorprop}
	&|xy - (x \pm \epsilon_1)(y \pm \epsilon_2)|  \\
	\leq & 3\max(X_{\max},Y_{\max}) \max(\epsilon_1 ,  \epsilon_2).
	\end{align}
	
\end{lemma}

Using the above four lemmas, we can obtain our main theoretical result in this section.  To concisely state the result, let us make the following definitions. Given an input channel vector $h:=\{h_{ij}\}\in\mathbb{R}^{K^2}$, let us use $v(h)^t_i$ to denote the variable $v_i$ at $t^{\rm th}$ iteration generated by WMMSE. Also let $H_{\min},H_{\max}>0$ denote the minimum and maximum channel strength, and let $V_{\min}>0$ be a given positive number. 
\begin{theorem}\label{thm:app_wmmse}
Suppose that WMMSE is randomly
initialized with $ (v^0_k)^2 \leq P_{\max},\ \sum_{i=1}^K v(h)^0_i \geq V_{\min} ,$ 
and it is executed for $T$ iterations.  Define the following set of `admissible' channel realizations
$$\mathcal{H} := \bigg\{h \mid H_{\min} \leq |h_{jk}| \leq H_{\max}, \forall j,k, \; {\sum_{i=1}^K v(h)^t_i \geq { V_{\min}}}, \forall t  \bigg\}.$$
Given $\epsilon>0$, there exists a neural network with $h\in\mathbb{R}^{K^2}$ {and $v^0\in\mathbb{R}_{+}^{K}$} as input and  $NET(h, v^0)\in\mathbb{R}_{+}^{K}$ as output, with the following number of layers{\small
\begin{align*}
O\bigg(&T^2\log\bigg(\max\bigg(K,P_{\max},H_{\max},\frac{1}{\sigma},\frac{1}{H_{\min}},\frac{1}{P_{\min}}\bigg)\bigg)+T\log\bigg(\frac{1}{\epsilon}\bigg)\bigg)
\end{align*}}
and the following number of  ReLUs and binary units {\small
\begin{align*}
O\bigg(T^2K^2\log\bigg(\max\bigg(K,P_{\max},H_{\max},\frac{1}{\sigma},\frac{1}{H_{\min}},\frac{1}{P_{\min}}\bigg)\bigg) + TK^2\log\bigg(\frac{1}{\epsilon}\bigg)\bigg),
\end{align*}}
such that the relation  below holds true
\begin{equation}
\max_{h \in \mathcal{H}}\max_{i} |(v(h)_i^T)^2 - NET(h,v^0)_i| \leq \epsilon
\end{equation}
\end{theorem}
\begin{corollary}
{The above statement still holds if the WMMSE algorithm has a fixed initialization with $(v^0_k)^2 = P_{\max}$, and the neural network has $h$ as input.}
\end{corollary}
\begin{remark}
{The admissible set $\mathcal{H}$ is defined mainly to ensure that the set of channels lie in a compact set. Admittedly, the conditions that the channel realizations are lower bounded by $H_{\min}$, and that they generate WMMSE sequences that satisfy $\sum_{i=1}^K v(h)^t_i \geq V_{\min}$ are somewhat artificial, but they are needed to bound the error propagation. It is worth noting that for problem \eqref{eq:virtual_IA_1}, $v_i=0,\; i=1,\cdots, K$ is an isolated local maximum solution (i.e., the only local maximum solution within a small neighborhood). Therefore it can be proved that as long as the WMMSE algorithm (which belongs to the coordinate descent type method) is not initialized at $v^0_i=0,\; i=1,\cdots, K$, then with probability one it will not converge to it; see recent result in \cite{lee16}. This fact helps  justify the assumption that there exists $V_{\min}>0$ such that $\sum_{i=1}^K v(h)^t_i \geq V_{\min}$ for all $t$.}
\end{remark}

\begin{remark}
	It is worth mentioning that, other operations such as binary search, and finding the solution for a linear systems of equation: $ A x = b$, can also be approximated using similar techniques.  Therefore, many resource allocation algorithms other than WMMSE can be approximated following similar analysis steps as in Theorem \ref{thm:app_wmmse}, as long as they can be expressed as the composition of these `basic' operations such as multiplication, division, binary search, threshloding operations, etc. 
	
\end{remark}

\begin{remark}
The bounds  in Theorem \ref{thm:app_wmmse} provide an intuitive understanding of how the size of the network should be dependent on various system parameters. A key observation is that having a neural network with multiple layers is essential in achieving our rate bounds. Another observation is that the effect of the approximation error on the size of the network is rather minor [the dependency is in the order of $\mathcal{O}(\log(1/\epsilon))$]. However, we do want to point out that the numbers predicted by Theorem \ref{thm:app_wmmse} represent some upper bounds on the size of the network. In practice much smaller networks are often used to achieve the best tradeoff between computational speed and solution accuracy.

\end{remark}

\begin{figure}
	\begin{center}
	 
		{\includegraphics[width=0.6\linewidth]{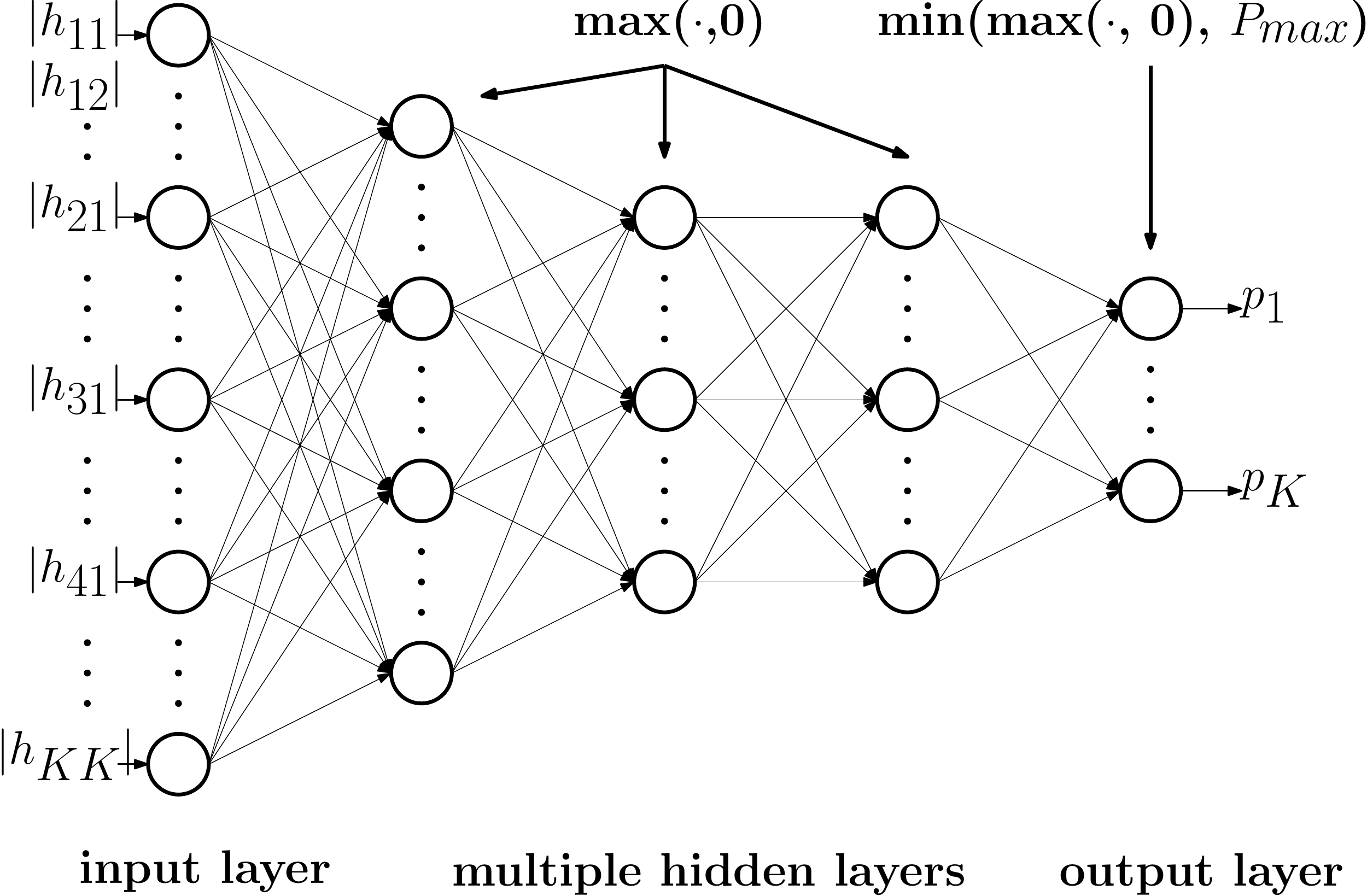}}
	
	\end{center}
 
	\caption{\footnotesize  The DNN structure used in this work. The fully connected neural network with one input layer, multiple hidden layers, and one output layer. The hidden layers use ReLU: $\mbox{max}(\cdot, 0)$ as the activation function, while the output layer use $\mbox{min}(\mbox{max}(\cdot, 0), P_{max})$ to incorporate the power constraint.}  
	\label{fig:structure}
 
\end{figure}

\section{System Setup} \label{ssec:1}
In this section, we describe in detail the DNN architectures as well as how the training and testing stages are performed. 

\noindent{\bf Network Structure.} Our proposed approach uses  a fully connected {neural network} with one input layer, multiple hidden layers, and one output layer as shown in Fig. \ref{fig:structure}. 	The input of the network is the magnitude of the channel coefficients $\{|h_{kj}|\}$, and the output of the network is the power allocation $\{p_k\}$. Further, we use ReLU as the activation function for the hidden layers activation. Additionally, to enforce the power constraint in \eqref{eq:sum_rate_IA} at the output of DNN, we also choose a special  activation function for the output layer, given below
\begin{equation}
y=\min(\max(x,0),P_{\max}).
\end{equation}

\noindent{\bf Data Generation.} The data is generated in the following manner. First, the channel realizations $\{h^{(i)}_{kj}\}$ are generated following certain distributions (to be specified in Section \ref{sec:numerical}), where the superscript $(i)$ is used to denote the index of the training sample. For simplicity we fix $P_{\max}$ and $\sigma_k$ for all $k$.  Then, for each tuple $(P_{\max}, \{\sigma_k\}, \{|h^{(i)}_{kj}|\})$, we generate the corresponding optimized  power vectors $\{p^{(i)}_k\}$ by running  WMMSE, with $v^0_k=\sqrt{P_{\max}}, \; \forall~k$ as initialization, and with $\mbox{obj}_{\rm new}-\mbox{obj}_{\rm old}<10^{-5}$ or $\mbox{number of iterations} > 500$ as termination criteria. The tuple $(\{|h^{(i)}_{kj}|\}, \{p^{(i)}_k\})$ is referred to as  the $i$th training sample. Then we repeat the above process for multiple times to generate the entire training data set, as well as the validation data set. The validation set is used to perform the cross-validation, model selection and early stopping during the training stage. Typically, the size of the validation data set is small compared with that of the training set. We use $\mathcal{T}$ and $\mathcal{V}$ to collect the indices  for the training and validation sets, respectively.

\noindent{\bf Training  Stage.}  We use the entire training data set  $(\{|h^{(i)}_{kj}|\}, \{p^{(i)}_k\})_{i\in\mathcal{T}}$ to optimize the weights of the  neural network. The cost function we use is the {mean squared error} between the label $\{p^{(i)}_k\}$ and the output of the network. The optimization algorithm we use is an efficient implementation of mini-batch {stochastic} gradient descent called the \emph{RMSprop} algorithm,  which divides the gradient by a running average of its recent magnitude \cite{RMSprop}. We choose the decay rate to be 0.9 as suggested in \cite{RMSprop} and select the proper learning rate and batch size by cross-validation.  To further improve the training performance, we initialize the weights using the truncated normal distribution \footnote{We generate a variable from the truncated normal distribution in the following manner: First, generate a variable from standard normal distribution. Then if its absolute value is larger than 2, it is  dropped and re-generated.}. Furthermore, we divide the weights of each neuron by the square root of its number of inputs to normalize the variance of each neuron's output\cite{sqrt_fanin}.

{\noindent{\bf Testing  Stage.} In the testing stage, we first generate the channels  following {\it the same} distribution as the training stage. For each channel realization, we pass it through the trained network and collect the optimized power. Then we compute the resulting sum-rate of the power allocation generated by DNN and compare it with that obtained by the WMMSE. 
We also test the robustness of the learned models by generating channels following distributions that are {\it different} from the training stage, and evaluate the resulting performance.

\section{Numerical Results}\label{sec:numerical}

\subsection{Simulation Setup}
  
The proposed DNN approach is implemented in Python 3.6.0 with TensorFlow 1.0.0 on one computer node with two 8-core Intel Haswell processors, two Nvidia K20 Graphical Processing Units  (GPUs), and 128 GB of memory.  The GPUs are used in the training stage to reduce the training time, but are {\it not} used in the testing stage. To rigorously compare the computational performance of WMMSE with the proposed approach, we have implemented the WMMSE in both Python and C.  The Python implementation is used to compare the computational performance under the same platform while the C implementation is used to achieve the best computational efficiency of WMMSE.

We consider the following two different channel models:

\noindent{\bf Model 1: Gaussian IC.} Each channel coefficient is generated according to a standard normal distribution, i.e., Rayleigh fading distribution with zero mean and unit variance. Rayleigh fading is a reasonable channel model that has been widely used to simulate the performance of various resource allocation algorithms. In our experiment, we consider three different network setups with $K\in \{10, 20, 30\}$.

\noindent{\bf Model 2: IMAC.}
For practical consideration, a multi-cell interfering MAC (IMAC) model is considered with a total of $N$ cells and $K$ users. The cell radius $R$ (i.e., the half distance between centers of adjacent cells) is set to be $100$ meters. In each cell, one BS is placed at the center of the cell and the users are randomly and uniformly distributed in the area from the inner circle of radius $r$ to the cell boundary; see Fig. \ref{fig:IMAC} for an illustration of the network configuration. We assume that the cells all have the same number of users. The channel between each user and each BS is randomly generated according to a Rayleigh fading distribution with zero mean and variance $(200/d)^3L$, where $d$ denotes the distance between the BS and user, and $\mbox{10log10}(L)$ following a normal distribution with zero mean and variance $64$; see \cite{liao13admm} for a similar network setup. In our experiment, we consider five difference network scenarios, with $(N,K)\in\{(3,12), (3,18), (3,24), (7,28), (20,80)\}$.

\begin{figure}[htb]
	\begin{minipage}[b]{.48\linewidth}
		\centering
		\centerline{\includegraphics[width=\linewidth]{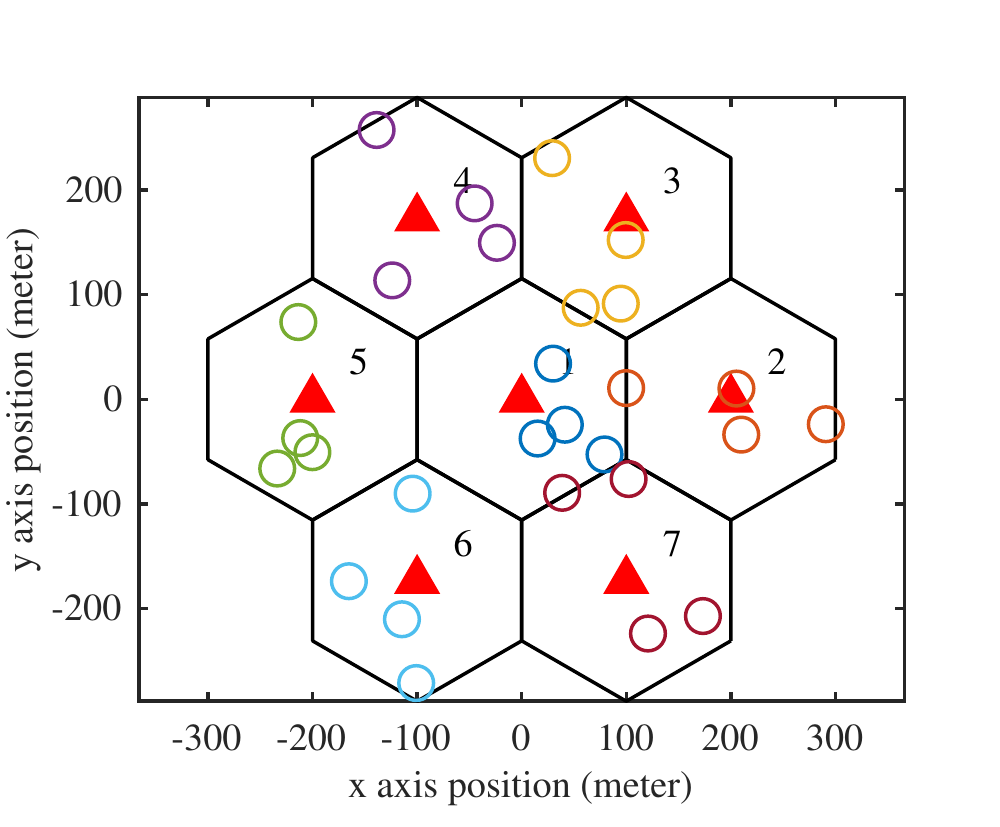}}
		\centerline{ \footnotesize (a) $R=100$m, $r=0$m} \medskip
	\end{minipage}
	\hfill
	\begin{minipage}[b]{0.48\linewidth}
		\centering
		\centerline{\includegraphics[width=\linewidth]{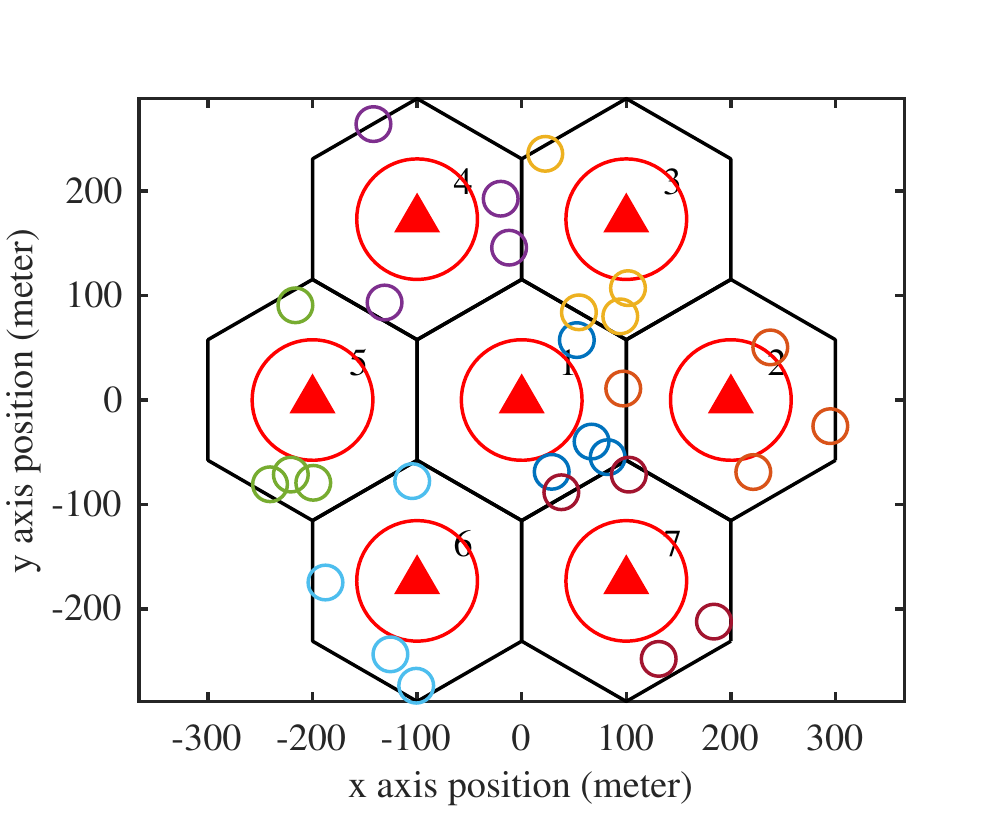}}
		\centerline{\footnotesize (b) $R=100$m, $r=50$m}\medskip
	\end{minipage}
 
	\caption{\footnotesize The network configuration of IMAC with 7 BSs and 28 users. Red triangles represent the BSs, black hexagons represent the boundaries of each cell, and the colored circles represent the user locations; a) shows a case that users located uniformly in the entire cell (with $r=0$); b) shows a case in which users can only locate $r=50m$ away from the cell center.}
	 
	\label{fig:IMAC}
\end{figure}

We mention that for each network scenario (i.e., IC/IMAC with a different number of BSs/users), we randomly generate one million realizations of the channel as the training data and ten thousand realizations of the channel as the validation data and testing data. {We use the distinctive random seeds for training, validation and testing data.}

\subsection{Parameter Selection}\label{sub:model}
 
We choose the following parameters for the neural network. For all our numerical results, we use a network with three hidden layers, one input layer, and one output layer. Each hidden layer contains $200$ neurons. The input to the network  is the set of channel coefficients, therefore the size of it depends on the channel models. More specifically, for  Gaussian IC the input size is $K^2$ while for the IMAC it is  $N\times K$. The output is the set of power allocations, therefore its size is $K$ for both Gaussian IC and IMAC.

To find parameters for the training algorithm, we perform cross-validation for different channel models as follows:

\noindent{\bf Model 1: Gaussian IC.} We study the impact of the batch size and learning rate on the {MSE} evaluated on the validation set,  as well as the total training time. Based on the result shown in Fig. \ref{fig:Cross}, we choose the batch size to be $1000$ and the initial learning rate to be $0.001$. Then we choose to gradually decrease the learning rate when the validation error does not decrease.

\noindent{\bf Model 2: IMAC.}
The parameter selection process is the same as that in the Gaussian IC channel.

\begin{figure}[htb]
	\begin{minipage}[b]{.48\linewidth}
		\centering
		\centerline{\includegraphics[width=\linewidth]{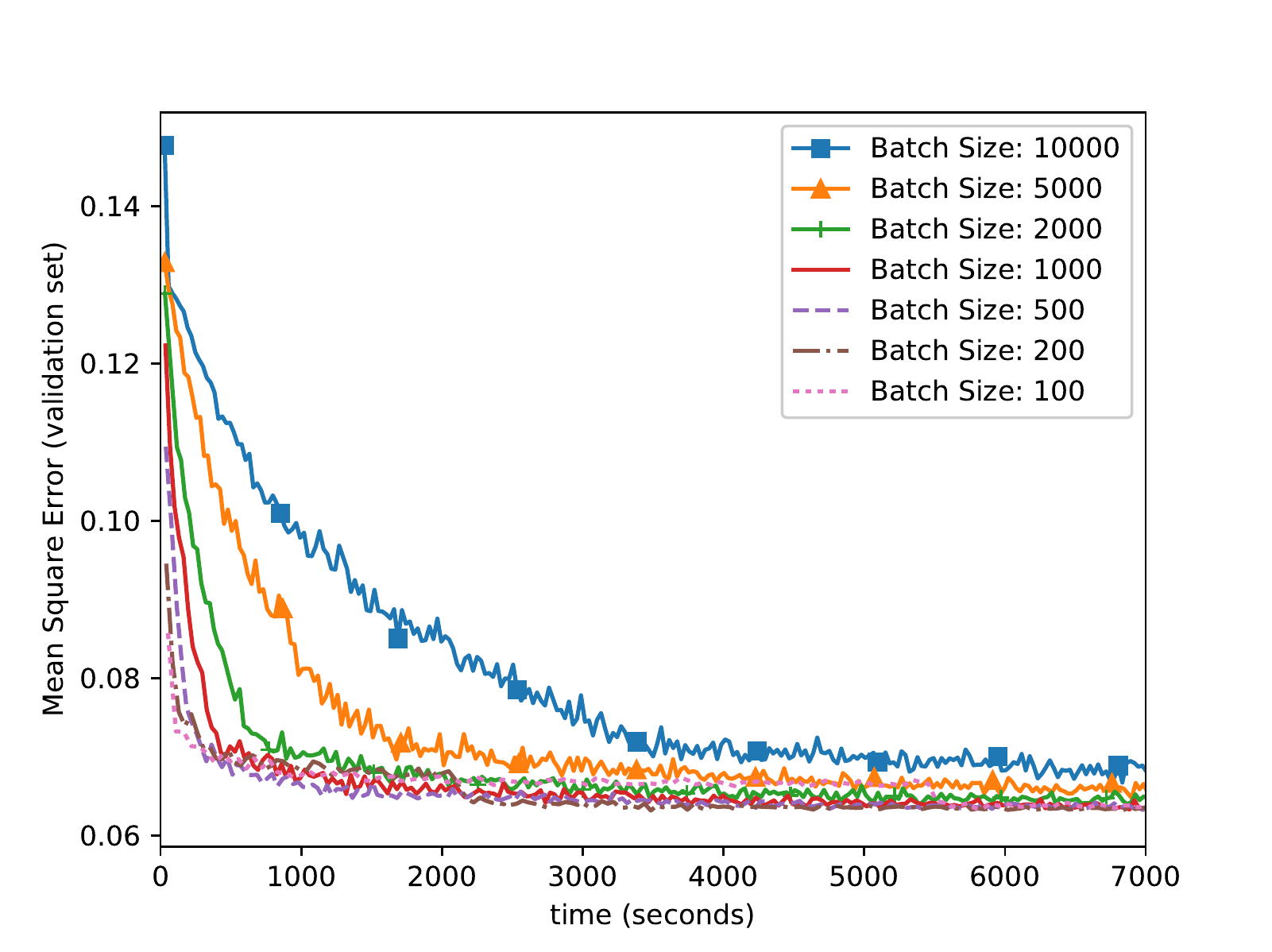}}
		\centerline{ \footnotesize (a) Batch size selection} \medskip
	\end{minipage}
	\hfill
	\begin{minipage}[b]{0.48\linewidth}
		\centering
		\centerline{\includegraphics[width=\linewidth]{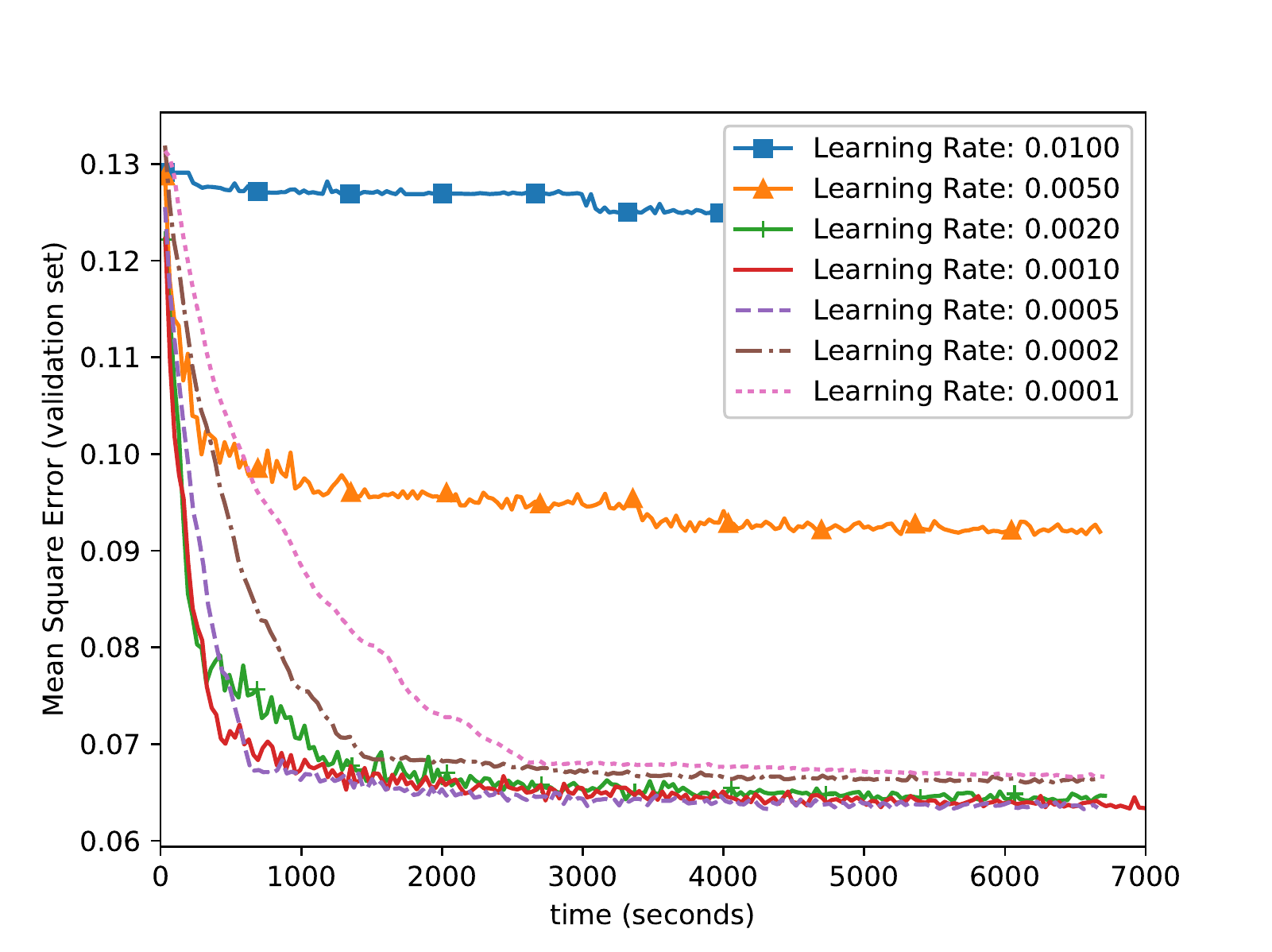}}
		\centerline{\footnotesize (b) Learning rate selection}\medskip
	\end{minipage}
	
	\caption{\footnotesize Parameter Selection for Gaussian IC case, $K=30$, where the MSEs are evaluated on the validation set. 
		Larger batch size leads to slower convergence, while smaller batch size incurs unstable convergence behavior. 
		Larger learning rate leads to a higher validation error, while the lower learning rate leads to slower convergence.
	}
	\label{fig:Cross}
\end{figure}

\subsection{Sum-Rate Performance}
We evaluate the sum-rate  performance of the DNN-based approach in the testing stage compared to the following schemes: 1) the WMMSE; 2) the random power allocation strategy, which simply generates the power allocation as:  $p_k\sim \mbox{Uniform}(0,P_{\max})$, $\forall~k$; 3) the maximum power allocation: $p_k=P_{\max}$, $\forall~k$; The latter two schemes serve as heuristic baselines. 
The simulation results are shown in Fig. \ref{fig:CDF} for both the Gaussian IC and the IMAC. Each curve in the figure represents the result obtained by averaging over 10, 000 randomly generated testing data points. It is observed that the sum-rate performance of DNN is very close to that of the WMMSE, while significantly outperforming the other two baselines. 
It is worth noting that for both Gaussian IC and IMAC, we have observed that the ``optimized" allocation strategies obtained by WMMSE are binary in most cases (i.e., $p_k$ takes the value of either $0$ or $P_{\max}$). Therefore, we also discretize the prediction to binary variables  to increase the accuracy. Specifically, for each output $y$ of our output layer, we {round the prediction up} to one if $y>0.5$, and {round down it} to zero if $y<0.5$.

\begin{figure}[htb]
	    
	\begin{minipage}[b]{.48\linewidth}
		\centering
		\centerline{\includegraphics[width=\linewidth]{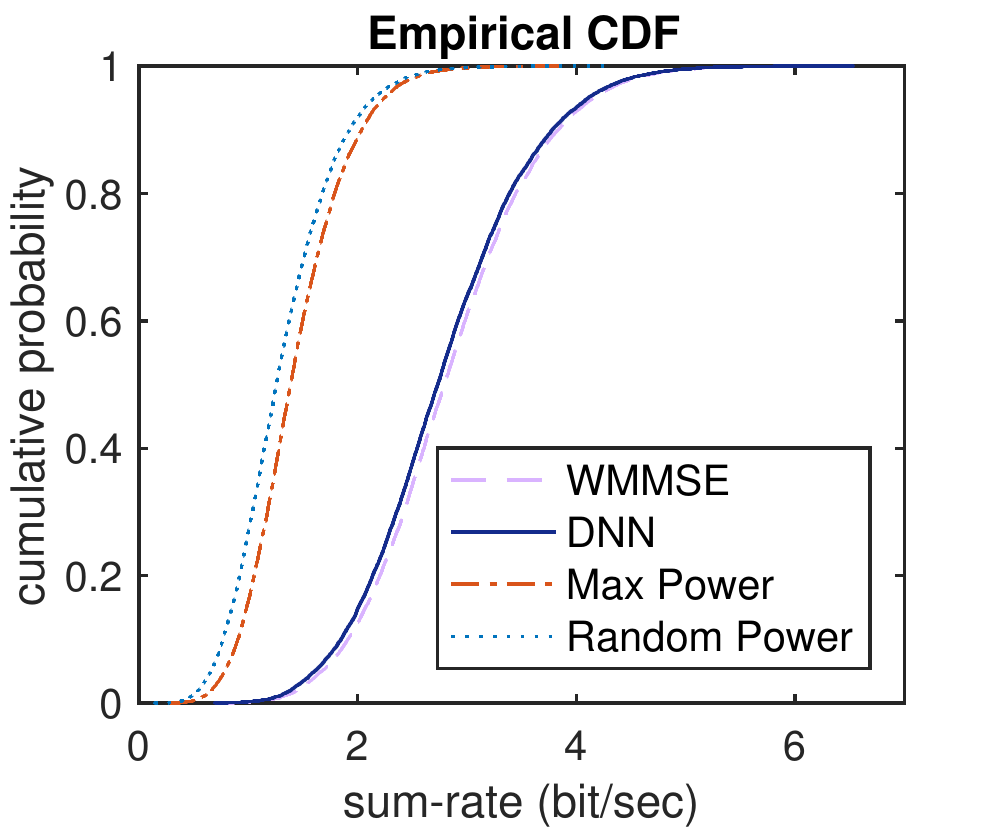}}
 
		\centerline{ \footnotesize (a) Gaussian IC Channel 
		} \medskip

	\end{minipage}
	\hfill
	\begin{minipage}[b]{0.48\linewidth}
		\centering
		\centerline{\includegraphics[width=\linewidth]{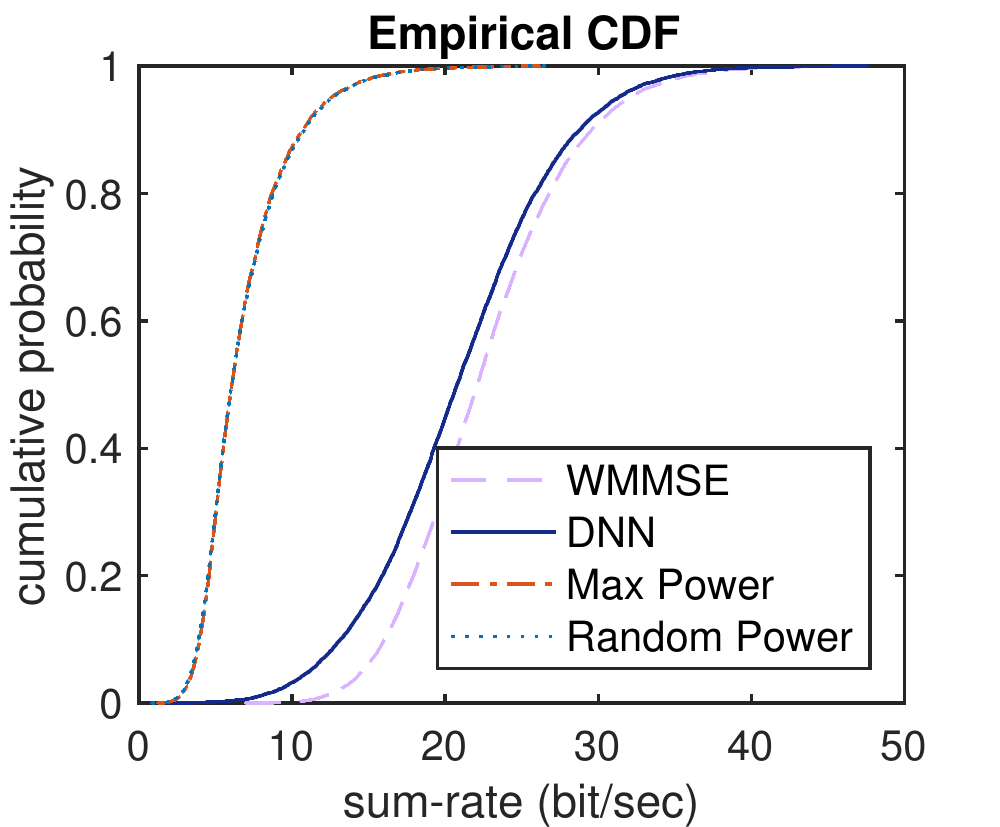}}
 
		\centerline{\footnotesize (b) IMAC Channel }\medskip
	\end{minipage}
	    
	\caption{\footnotesize The cumulative distribution function (CDF) that describes the rates achieved by different algorithms, over $10,000$ randomly generated testing data points, where a) shows the Gaussian IC with $K=10$, b) shows the IMAC with $N=3$ and $K=24$.
		{Compared with WMMSE, on average 98.33\% (resp. 96.82\%) accuracy is achieved for Gaussian IC Channel (resp. IMAC).}   }
	\label{fig:CDF}
	    
\end{figure}

\subsection{Scalability Performance}
 
In this subsection, we demonstrate the scalability of the proposed DNN approach when the size of the wireless network is increased. 
The average achieved sum-rate performance (averaged using $10,000$ test samples) and the percentage of achieved sum-rate of DNN over that of the WMMSE are presented  in TABLE \ref{table:IC} (for the IC) and TABLE \ref{table:IMAC} (for the IMAC).  
It can be seen that our proposed method achieves good scalability for prediction accuracy and computational efficiency. For example, compared with the IMAC WMMSE (which is implemented in C) with $N=20$ and $K=80$, the DNN obtains $90.36\%$ sum-rate while achieving {{\it more than one hundred times} (two orders of magnitude)} speed up. Additionally, in  Fig. \ref{fig:Histogram}, we show the distribution of the sum-rates over the entire test dataset. It is observed that the DNN approach gives a fairly good approximation of the entire rate profile generated by the WMMSE.

\begin{table*}[t]
	 
	\tabcolsep 0pt \caption{ Sum-Rate and Computational Performance for Gaussian IC}

	\begin{center}
	\def\temptablewidth{1\textwidth}
	{\rule{\temptablewidth}{1pt}}
	\begin{tabular*}{\temptablewidth}{@{\extracolsep{\fill}}c|cccc|cccc}
		&\multicolumn{3}{c}{average sum-rate (bit/sec.)}  &&\multicolumn{4}{c}{total CPU time (sec.)}\\
		{\# of users (K)} & DNN	& WMMSE      & DNN/WMMSE &  & DNN  & WMMSE (Python)  &WMMSE (C)       & DNN/WMMSE(C)               \\\hline \hline
		10                 & 2.770   &   2.817        & 98.33\%   &  &  0.047    &  17.98 &   0.355    &  13.24\%    \\
		20                 &3.363      &   3.654      & 92.04\%   &  &   0.093    &59.41 &    1.937     & 4.80\%      \\
		30                 & 3.498     &   4.150        & 84.29\%   &  &  0.149    & 122.91 &   5.103       &2.92\%      \\
		
	\end{tabular*}
	{\rule{\temptablewidth}{1pt}}
\end{center}
	\label{table:IC}
	 
\end{table*}

\begin{table*}[t]
	 
	\tabcolsep 0pt \caption{ Sum-Rate and Computational Performance  for IMAC}
	 
	\begin{center}
		\def\temptablewidth{1\textwidth}
		{\rule{\temptablewidth}{1pt}}
		\begin{tabular*}{\temptablewidth}{@{\extracolsep{\fill}}c|cccc|cccc}
			{\# of base stations }&\multicolumn{3}{c}{average sum-rate (bit/sec.)}&&\multicolumn{4}{c}{total CPU time (sec.)}\\
			{and users  (N,K)}	& DNN & WMMSE       & DNN/WMMSE   && DNN   &WMMSE (MATLAB) &WMMSE(C) & DNN/WMMSE (C)   \\\hline \hline
			$( 3,12)$             &  17.722 &18.028& 98.30\%&&0.021& 22.33&0.27 &7.78\% \\
			$( 3,18)$            &  20.080 &20.606& 97.45\%&&0.022& 42.77& 0.48&4.58\% \\
			$( 3,24)$            &  21.928 &22.648& 96.82\%&&0.025&  67.59& 0.89&2.81\% \\
			$( 7,28)$                 &  33.513 &35.453& 94.53\%&&0.038&  140.44& 2.41&1.58\% \\
			$(20,80)$                 &  79.357 &87.820& 90.36\%&&0.141& 890.19& 23.0&0.61\% \\
			
		\end{tabular*}
		{\rule{\temptablewidth}{1pt}}
	\end{center}
	\label{table:IMAC}
 
\end{table*}

\begin{figure}[!htb]
	\minipage{0.3333\textwidth}
	\includegraphics[width=\linewidth]{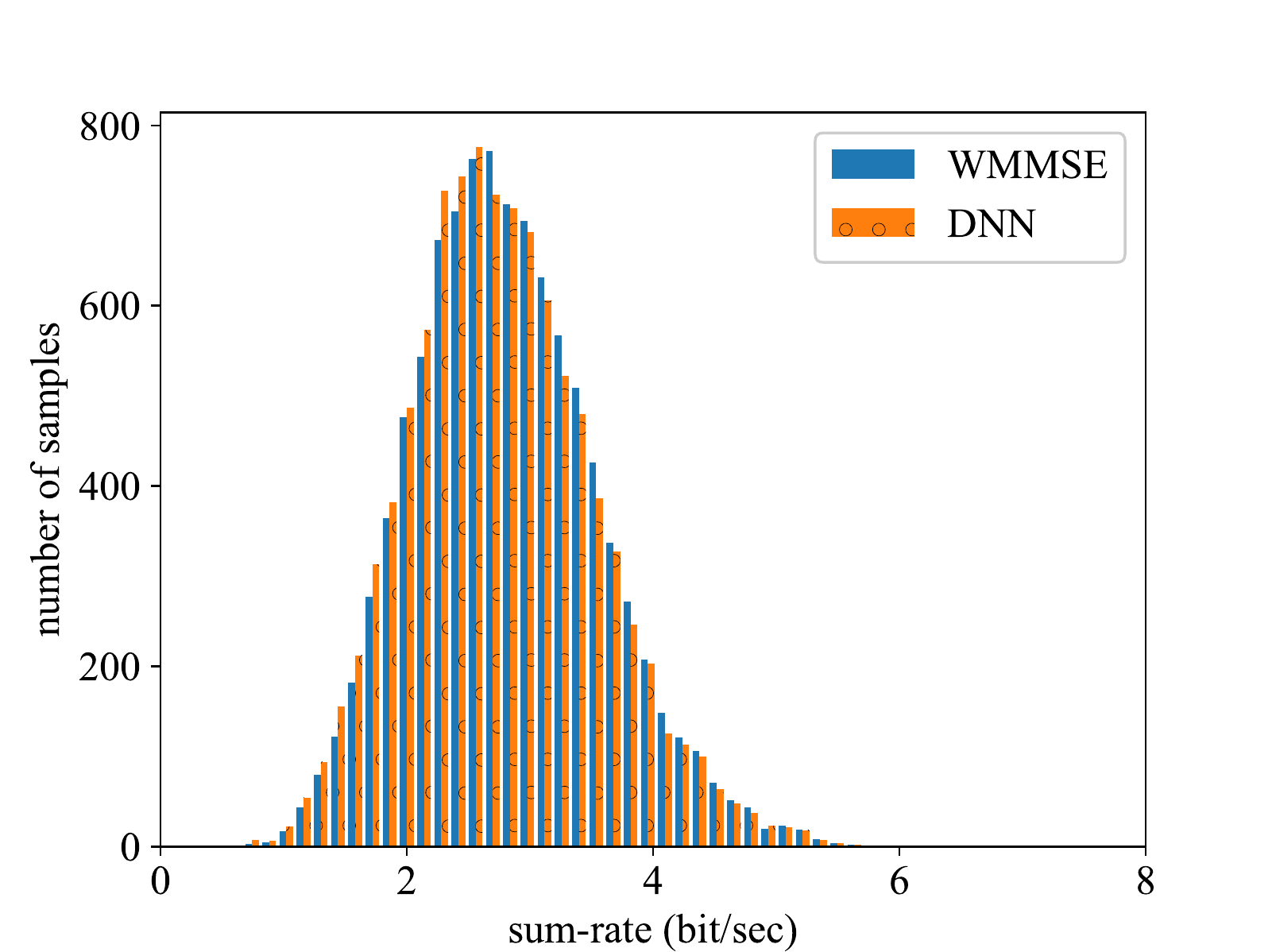}
	\centerline{ \footnotesize (a) K=10}
	\endminipage
	\minipage{0.3333\textwidth}
	\includegraphics[width=\linewidth]{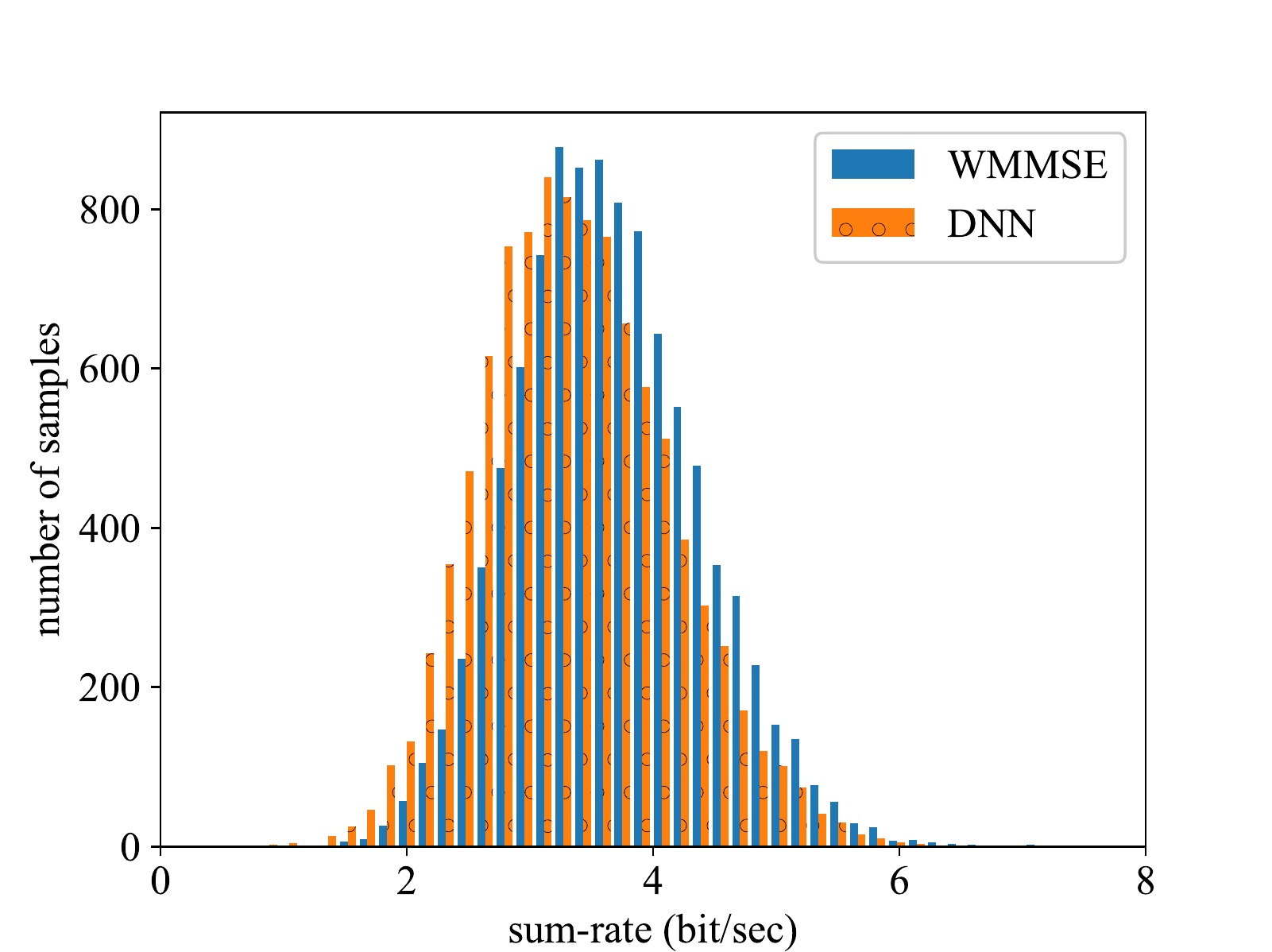}
	\centerline{ \footnotesize (b) K=20}
	\endminipage\hfill
	\minipage{0.3333\textwidth}%
	\includegraphics[width=\linewidth]{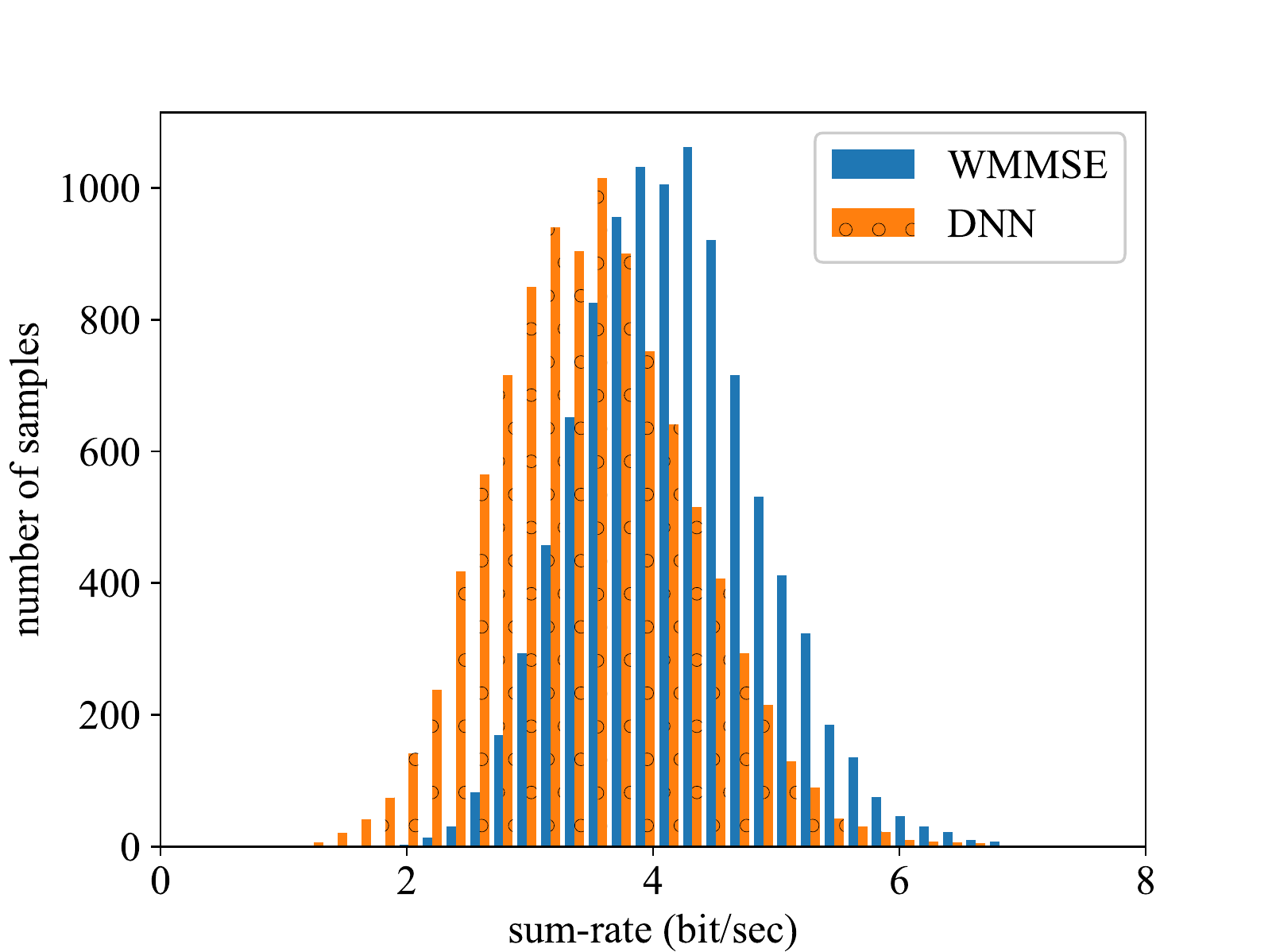}
	\centerline{ \footnotesize (c) K=30}
	\endminipage
	\caption{\footnotesize Distributions of the proposed DNN approach and the WMMSE in various scenarios for Gaussian IC Channel.} 
	 
	\label{fig:Histogram}
\end{figure}

{Moreover, from Table \ref{table:IC} and Table \ref{table:IMAC}, we observe that the DNN performs better in the IMAC setting,  achieving lower computational time and higher relative sum-rate (relative to WMMSE). 
The improved relative computational performance in IMAC is due to the fact that, for the same number of users $K$, the IMAC has $K\times N$ different input channel coefficients while the IC has $K^2$. Therefore when $N\ll K$, the DNN built for IMAC has a smaller size compared with that of the IC. On the contrary, the computation steps for WMMSE are the same for IMAC and IC, where at each iteration steps 7-9 in Fig. \ref{fig:pseudo_code1} are executed for each user $K$. These facts combined together result in better real-time computational efficiency for DNN in the IMAC setting. 
On the other hand, the improved relative sum rate performance in IMAC is mainly due to the fact that the IMAC channels are more `structured' (which are generated according to actual users locations and shadow fading) compared with the IC channels  (which are simply randomly generated from Gaussian distribution). This means that the training and testing samples for the  IMAC  are  closely related to each other,  leading to better generalization capabilities of the DNN model in the testing stage.  }

\subsection{Generalization Performance}

In the previous subsection, we have demonstrated the superiority of our DNN approach, when the distribution of the testing data is the same as that of the training data. However, in practice, there could be significant model mismatch between the training and the testing dataset. For example, the number of users in the system is constantly changing, but it is impossible to train a different neural network for each network configuration. An interesting question is: Whether the model learned by the DNN is able to adapt to such training-testing model mismatch? In this {subsection}, we investigate the {\it generalization} performance of the proposed DNN approach.

\noindent{\bf Model 1: Gaussian IC.}
We study the scenario in which only $K/2$ users are present in the testing, while the DNN is trained with $K$ users (we name this case  ``half user" in TABLE \ref{table:half}). The purpose is to understand the generalization capability of the trained model when input dimension is different. Somewhat surprisingly, we observe that in this case, the DNN still performs well. This result suggests that it is not necessary to train one DNN for each network configuration, which makes the proposed approach very attractive for practical use. Note that in TABLE \ref{table:half} the input size for WMMSE is reduced from $K^2$ to $K^2/4$ while the input size for DNN is still $K^2$ (the unused inputs are replaced by zeros). Therefore compared with the results in TABLE \ref{table:IC}, we can observe that the computational time for WMMSE has been significantly reduced.

\begin{table*}[t]
	\tabcolsep 0pt \caption{ Sum-Rate and Computational Performance for Gaussian IC (Half User Case)}
 
\begin{center}
	\def\temptablewidth{1\textwidth}
	{\rule{\temptablewidth}{1pt}}
	\begin{tabular*}{\temptablewidth}{@{\extracolsep{\fill}}c|cccc|cccccc}
		&\multicolumn{3}{c}{average sum-rate (bit/sec.)} &&\multicolumn{4}{c}{total CPU time (sec.)}\\
		{\# of users (K)} 	&DNN &WMMSE &DNN/WMMSE &&  DNN & WMMSE (Python) & WMMSE (C)& DNN/WMMSE(C)\\ \hline \hline
		10                   &      2.048   &   2.091    &   97.94\%&    &    0.044  & 5.58 & 0.074   &  59.45\%    \\
		20                   &      2.595 &    2.817     &    92.12\%  &  &   0.087 & 18.36&0.346   &  25.14\%     \\
		30                   &      2.901  &       3.311  &    87.62\%   &  &   0.157  &36.98& 0.908  &  17.29\%    
	\end{tabular*}
	{\rule{\temptablewidth}{1pt}}
\end{center}
	\label{table:half}
 
\end{table*}

\noindent{\bf Model 2: IMAC.}
For the IMAC model with fixed number of base stations $N$, there could be  many choices of network parameters, such as $R, r$ and $K$. We study the cases that the DNN is trained with $R=100$ meters and $r=0$ meters,  and apply the resulting model to networks with with $R$ ranging from 100 to 500 meters, and $r$ from 0 to 99 meters (note that $r=99$ meters and $R=100$ meters represents the challenging case where all the users are on the cell boundary). The results are shown in TABLE \ref{table:IMAC_gen1} and TABLE \ref{table:IMAC_gen2}, from which we can conclude that our model generalizes relatively well when the testing network configurations are sufficiently close  to those used in the training. Further, as expected, when the testing models are very different, e.g., $R=500$ meters, the performance degrades. 

\begin{table*}[t]
	 
	\tabcolsep 0pt \caption{ Sum-Rate and Computational Performance for IMAC (With Different Cell Radius $R$)}
 
	\begin{center}
		\def\temptablewidth{1\textwidth}
		{\rule{\temptablewidth}{1pt}}
 
		\begin{tabular*}{\temptablewidth}{@{\extracolsep{\fill}}c|cccc|cccc}
			{\# of base stations }&\multicolumn{4}{c|}{sum-rate performance}&\multicolumn{4}{c}{computational time performance}\\
			{and users  (N,K)}	& R=100m  & R=200m   & R=300m    & R=500m& R=100m  & R=200m   & R=300m    & R=500m \\\hline \hline
			$(3,12)$          & 98.30\% & 99.51\%  &99.62\% & 98.91\%&7.78\%  &19.16\% &19.64\% & 19.77\% \\
			$(3,18)$            &  97.45\% &98.27\% & 97.80\% &95.58\%&4.58\% &13.06\% &  14.67\% & 13.34\% \\
			$(3,24)$          & 96.82\%  & 96.77\%&93.83\%  &87.32\% &2.81\% &10.68\% & 12.44\% & 10.83\% \\
			$(7,28)$               &94.53\%   & 97.19\%& 95.41\% &90.22\%&1.58\% & 5.84\% &   8.68\% &  8.69\%\\
			$(20,80)$               & 90.36\%  &86.27\% & 65.55\% &43.48\%&0.61\% &1.75\% & 3.21\% & 3.31\% \\
		\end{tabular*}
		{\rule{\temptablewidth}{1pt}}
	 
	\end{center}
	\label{table:IMAC_gen1}
\end{table*}

\begin{table*}[t]
	 
	\tabcolsep 0pt \caption{ Sum-Rate and Computational Performance for IMAC (With Different Inner Circle Radius $r$)}
 
	\begin{center}
		\def\temptablewidth{1\textwidth}
		{\rule{\temptablewidth}{1pt}}
		\begin{tabular*}{\temptablewidth}{@{\extracolsep{\fill}}c|cccc|cccc}
			{\# of base stations }&\multicolumn{4}{c|}{sum-rate performance}&\multicolumn{4}{c}{computational time performance}\\
			{and users  (N,K)}	& $r$=0m  & $r$=20m   & $r$=50m    & $r$=99m & $r$=0m  & $r$=20m   & $r$=50m    & $r$=99m \\\hline \hline
			$(3,12)$  & 98.30\% &98.29\%  & 98.38\% & 98.46\%& 7.78\%   & 8.71\%&11.23\% &14.68\%  \\
			$(3,18)$ &  97.45\% &96.99\%  &96.87\%   &96.64\%& 4.58\% &4.70\%  &6.13\% &8.42\%  \\
			$(3,24)$ &  96.82\%& 95.85\%& 95.41\% &94.65\%& 2.81\%  & 3.36\% &4.49\% &6.09\%  \\
			$(7,28)$ & 94.53\%   &93.80\%  & 93.47\%  &92.55\%&1.58\%  & 2.14\% &2.69\% &3.89\%  \\
			$(20,80)$ & 90.36\% & 88.51\%&  86.04\%&82.05\%&0.61\%  &  0.85\%& 1.12\%&1.45\%  \\
		\end{tabular*}
		{\rule{\temptablewidth}{1pt}}
	\end{center}
	\label{table:IMAC_gen2}
 
\end{table*}

\subsection{Real Data Performance}
In this section, we further validate our approach using real data. We consider the interference management problem over the very-high-bit-rate digital subscriber line (VDSL) channels, which can be cast as a power control problem over an interference channel \cite{hong12survey}.
The data is collected by France Telecom R$\&$D; see \cite{karipidis2006crosstalk, karipidis2005experimental} and the references therein. The measured lengths of the VLSI lines were 75 meters, 150 meters, and 300 meters, each containing $28$ loops, and the bandwidth is up to 30 MHz. For each length, all 378 (28 choose 2) crosstalk channels of the far-end crosstalk (FEXT) and near-end crosstalk (NEXT) were measured, for a total of over 3000 crosstalk channels. There is a total of $6955$ channel measurements\cite{karipidis2006crosstalk, karipidis2005experimental}.   

\noindent{\bf Problem setup.} We model the VDSL channel using a $28$-user IC, and we use the magnitude of each channel coefficient to compute the power allocation (by using WMMSE). 
The entire 6955 channel  samples are divided into  5000 samples of {\it validation} set and 1955 samples of {\it testing} set. The  {\it training} data set is {computer-generated following channel statistics learned} from {the} validation set. Specifically, we first calculate the means and variances for both the direct channel and the interfering channels {in} the validation data set, denoted as $(m_d, \sigma^2_d)$, $(m_i, \sigma^2_i)$, respectively. Then we randomly generate 28 direct channel coefficients independently using $\mathcal{N}(m_d, \sigma^2_{d})$, and $28\times 27$ interfering channels from $\mathcal{N}(m_i, \sigma^2_{i})$ to form one training data sample. 
Repeating the above  process, we generate $50,000$ training samples for each measured length.

Using the training, validation and testing data set as described above,  we perform the training and testing following the procedures outlined in {Section \ref{ssec:1}}. {Since the channel coefficients are relatively small compared to our previous random generated channel, we set the environment noise variance to $0.001$ to make sure that WMMSE  gives meaningful output. 

\noindent{\bf Numerical results.} The results are shown in TABLE \ref{table:real}. Despite the fact that the total available data set is quite limited, the proposed DNN approach can still achieve relatively high sum-rate performance on the measured testing data set. This result is encouraging, and it further validates our `learning to optimize' approach. 

\begin{table*}[t]
	\tabcolsep 0pt \caption{Sum-Rate and Computational Performance for Measured VDSL Data}
	\begin{center}
		\def\temptablewidth{1\textwidth}
		{\rule{\temptablewidth}{1pt}}
		\begin{tabular*}{\temptablewidth}{@{\extracolsep{\fill}}c|cccc|cccccc}
			&\multicolumn{3}{c}{average sum-rate (bit/sec.)} &&\multicolumn{3}{c}{total CPU time (sec.)}\\
			{(length, type)} 	&DNN &WMMSE &DNN/WMMSE &&  DNN& WMMSE (C)& DNN/WMMSE(C)\\ \hline \hline
		( 75, FEXT) & 203.61  & 203.69 & 99.96\%&&0.027&0.064&42.18\%\\  
		(150, FEXT) & 135.61  & 136.38 & 99.43\%&&0.026&0.051&50.98\%\\
		(300, FEXT) & 50.795 & 51.010 & 99.58\%&&0.026&0.045&57.78\%\\ \hline
		( 75, NEXT) & 183.91  & 184.19 & 99.85\%&&0.026&0.824&3.16\%\\
		(150, NEXT) & 109.64 & 111.52 & 98.31\%&&0.026&0.364&7.14\%\\
		(300, NEXT) & 41.970 & 44.582  & 94.14\%&&0.026&0.471&5.52\%
		\end{tabular*}
		{\rule{\temptablewidth}{1pt}}
	\end{center}
	\label{table:real}
\end{table*}

\section{Perspectives and Future Work}

In this work, we design a new `learning to optimize' based framework for optimizing wireless resources. Our theoretical results indicate that it is possible to learn a well-defined optimization algorithm very well by using finite-sized deep neural networks. Our empirical results show that,  for the power control problems over either the IC or the IMAC channel, deep neural networks can be trained to well-approximate the behavior of the state-of-the-art algorithm WMMSE. In many aspects, our results are encouraging. The key message: {\it DNNs have great potential as computationally cheap surrogates of expensive optimization algorithms for quasi-optimal and real-time wireless resource allocation}. We expect that they could also be used in many other computationally expensive signal processing tasks with stringent real-time requirements.  However, the current work still represents a preliminary step towards understanding the capability of DNNs (or related learning algorithms) for this type of problems. There are many interesting questions to be addressed in the future, and some of them are listed below:
\begin{enumerate}
	\item How to further reduce the computational complexity of DNNs? 
	\item Is it possible to design specialized neural network architectures (beyond the fully connected architecture used in the present work) to further improve the real-time computation and testing performance?
		\item Is it possible to generalize our proposed approach to more challenging problems such as beamforming for IC/IMAC?
\end{enumerate}

{\footnotesize
\bibliographystyle{IEEEbib}
\bibliography{ref,biblio,icassp_dnn}

\begin{thebibliography}{10}

\bibitem{sun2017learning}
H.~Sun, X.~Chen, Q.~Shi, M.~Hong, X.~Fu, and N.~D. Sidiropoulos,
\newblock ``Learning to optimize: Training deep neural networks for wireless
  resource management,''
\newblock in {\em Signal Processing Advances in Wireless Communications
  (SPAWC), 2017 IEEE 18th International Workshop on}. IEEE, 2017.

\bibitem{hong12survey}
M.~Hong and Z.-Q. Luo,
\newblock ``Signal processing and optimal resource allocation for the
  interference channel,''
\newblock in {\em Academic Press Library in Signal Processing}. Academic Press,
  2013.

\bibitem{bjornson13}
E.~Bjornson and E.~Jorswieck,
\newblock ``Optimal resource allocation in coordinated multi-cell systems,''
\newblock {\em Foundations and Trends in Communications and Information
  Theory}, vol. 9, 2013.

\bibitem{yu02a}
W.~Yu, G.~Ginis, and J.~M. Cioffi,
\newblock ``Distributed multiuser power control for digital subscriber lines,''
\newblock {\em IEEE Journal on Selected Areas in Communications}, vol. 20, no.
  5, pp. 1105--1115, 2002.

\bibitem{scutari08a}
G.~Scutari, D.~P. Palomar, and S.~Barbarossa,
\newblock ``Optimal linear precoding strategies for wideband noncooperative
  systems based on game theory -- part {I}: Nash equilibria,''
\newblock {\em IEEE Transactions on Signal Processing}, vol. 56, no. 3, pp.
  1230--1249, 2008.

\bibitem{Schmidt09}
D.~Schmidt, C.~Shi, R.~Berry, M.~Honig, and W.~Utschick,
\newblock ``Distributed resource allocation schemes,''
\newblock {\em IEEE Signal Processing Magazine}, vol. 26, no. 5, pp. 53 --63,
  2009.

\bibitem{papand09}
J.~Papandriopoulos and J.~S. Evans,
\newblock ``{SCALE}: A low-complexity distributed protocol for spectrum
  balancing in multiuser {DSL} networks,''
\newblock {\em IEEE Transactions on Information Theory}, vol. 55, no. 8, pp.
  3711--3724, 2009.

\bibitem{shi11WMMSE_TSP}
Q.~Shi, M.~Razaviyayn, Z.-Q. Luo, and C.~He,
\newblock ``An iteratively weighted {MMSE} approach to distributed sum-utility
  maximization for a {MIMO} interfering broadcast channel,''
\newblock {\em IEEE Transactions on Signal Processing}, vol. 59, no. 9, pp.
  4331--4340, 2011.

\bibitem{kim11}
S.-J. Kim and G.~B. Giannakis,
\newblock ``Optimal resource allocation for {MIMO} {Ad Hoc} {Cognitive Radio
  Networks},''
\newblock {\em IEEE Transactions on Information Theory}, vol. 57, no. 5, pp.
  3117 --3131, 2011.

\bibitem{luo10SDPMagazine}
Z.-Q. Luo, W.-K. Ma, A.M.-C. So, Y.~Ye, and S.~Zhang,
\newblock ``Semidefinite relaxation of quadratic optimization problems,''
\newblock {\em IEEE Signal Processing Magazine}, vol. 27, no. 3, pp. 20 --34,
  2010.

\bibitem{liu13deflation}
Y.-F Liu, Y.-H. Dai, and Z.-Q. Luo,
\newblock ``Joint power and admission control via linear programming
  deflation,''
\newblock {\em IEEE Transactions on Signal Processing}, vol. 61, no. 6, pp.
  1327 --1338, 2013.

\bibitem{Matskani08}
E.~Matskani, N.~Sidiropoulos, Z.-Q. Luo, and L.~Tassiulas,
\newblock ``Convex approximation techniques for joint multiuser downlink
  beamforming and admission control,''
\newblock {\em IEEE Transactions on Wireless Communications}, vol. 7, no. 7,
  pp. 2682 --2693, 2008.

\bibitem{hong12sparse}
M.~Hong, R.~Sun, H.~Baligh, and Z.-Q. Luo,
\newblock ``Joint base station clustering and beamformer design for partial
  coordinated transmission in heterogenous networks,''
\newblock {\em IEEE Journal on Selected Areas in Communications.}, vol. 31, no.
  2, pp. 226--240, 2013.

\bibitem{baligh2014cross}
H.~Baligh, M.~Hong, W.-C. Liao, Z.-Q. Luo, M.~Razaviyayn, M.~Sanjabi, and
  R.~Sun,
\newblock ``Cross-layer provision of future cellular networks: A {WMMSE}-based
  approach,''
\newblock {\em IEEE Signal Processing Magazine}, vol. 31, no. 6, pp. 56--68,
  Nov 2014.

\bibitem{yu02b}
W.~Yu and J.~M. Cioffi,
\newblock ``{FDMA} capacity of {G}aussian multiple-access channel with isi,''
\newblock {\em IEEE Transactions on Communications}, vol. 50, no. 1, pp.
  102--111, 2002.

\bibitem{lecun2015deep}
Y.~LeCun, Y.~Bengio, and G.~Hinton,
\newblock ``Deep learning,''
\newblock {\em Nature}, vol. 521, no. 7553, pp. 436--444, 2015.

\bibitem{gregor2010learning}
K.~Gregor and Y.~LeCun,
\newblock ``Learning fast approximations of sparse coding,''
\newblock in {\em Proceedings of the 27th International Conference on Machine
  Learning (ICML-10)}, 2010, pp. 399--406.

\bibitem{Beck:2009:FIS:1658360.1658364}
A.~Beck and M.~Teboulle,
\newblock ``A fast iterative shrinkage-thresholding algorithm for linear
  inverse problems,''
\newblock {\em SIAM Journal on Imgaging Science}, vol. 2, no. 1, pp. 183 --
  202, 2009.

\bibitem{hershey2014deep}
J.~R. Hershey, J.~Le Roux, and F.~Weninger,
\newblock ``Deep unfolding: Model-based inspiration of novel deep
  architectures,''
\newblock {\em arXiv preprint arXiv:1409.2574}, 2014.

\bibitem{sprechmann2013supervised}
P.~Sprechmann, R.~Litman, T.~B. Yakar, A.~M. Bronstein, and G.~Sapiro,
\newblock ``Supervised sparse analysis and synthesis operators,''
\newblock in {\em Advances in Neural Information Processing Systems}, 2013, pp.
  908--916.

\bibitem{andrychowicz2016learning}
M.~Andrychowicz, M.~Denil, S.~Gomez, M.~Hoffman, D.~Pfau, T.~Schaul, and
  N.~de~Freitas,
\newblock ``Learning to learn by gradient descent by gradient descent,''
\newblock in {\em Advances in Neural Information Processing Systems}, 2016, pp.
  3981--3989.

\bibitem{li2016learning}
K.~Li and J.~Malik,
\newblock ``Learning to optimize,''
\newblock {\em arXiv preprint arXiv:1606.01885}, 2016.

\bibitem{o2016recurrent}
T.~J. O'Shea, T.~C. Clancy, and R.~McGwier,
\newblock ``Recurrent neural radio anomaly detection,''
\newblock {\em arXiv preprint arXiv:1611.00301}, 2016.

\bibitem{farsad2017detection}
N.~Farsad and A.~Goldsmith,
\newblock ``Detection algorithms for communication systems using deep
  learning,''
\newblock {\em arXiv preprint arXiv:1705.08044}, 2017.

\bibitem{west2017deep}
N.~E. West and T.~J. O'Shea,
\newblock ``Deep architectures for modulation recognition,''
\newblock in {\em Dynamic Spectrum Access Networks (DySPAN), 2017 IEEE
  International Symposium on}. IEEE, 2017, pp. 1--6.

\bibitem{o2017deep}
T.~J. O'Shea, T.~Erpek, and T.~C. Clancy,
\newblock ``Deep learning based {MIMO} communications,''
\newblock {\em arXiv preprint arXiv:1707.07980}, 2017.

\bibitem{samuel2017deep}
N.~Samuel, T.~Diskin, and A.~Wiesel,
\newblock ``Deep mimo detection,''
\newblock {\em arXiv preprint arXiv:1706.01151}, 2017.

\bibitem{luo08a}
Z-.Q. Luo and S.~Zhang,
\newblock ``Dynamic spectrum management: Complexity and duality,''
\newblock {\em IEEE Journal of Selected Topics in Signal Processing}, vol. 2,
  no. 1, pp. 57--73, 2008.

\bibitem{verdu98}
S.~Verdu,
\newblock {\em Multiuser Detection},
\newblock Cambridge University Press, Cambridge, UK, 1998.

\bibitem{bertsekas97}
D.~P. Bertsekas and J.~N. Tsitsiklis,
\newblock {\em Parallel and Distributed Computation: Numerical Methods, 2nd
  ed},
\newblock Athena Scientific, Belmont, MA, 1997.

\bibitem{hornik1989multilayer}
K.~Hornik, M.~Stinchcombe, and H.~White,
\newblock ``Multilayer feedforward networks are universal approximators,''
\newblock {\em Neural networks}, vol. 2, no. 5, pp. 359--366, 1989.

\bibitem{whydnn}
S.~Liang and R.~Srikant,
\newblock ``Why deep neural networks for function approximation?,''
\newblock {\em ICLR}, 2017.

\bibitem{lee16}
J.~D. Lee, M.~Simchowitz, M.~I. Jordan, and B.~Recht,
\newblock ``Gradient descent converges to minimizers,''
\newblock 2016,
\newblock Preprint, available at arXiv:1602.04915v1.

\bibitem{RMSprop}
G.~Hinton, N.~Srivastava, and K.~Swersky,
\newblock ``Lecture 6a overview of mini--batch gradient descent,''
\newblock {\em Coursera Lecture slides https://class. coursera.
  org/neuralnets-2012-001/lecture,[Online]}, 2012.

\bibitem{sqrt_fanin}
X.~Glorot and Y.~Bengio,
\newblock ``Understanding the difficulty of training deep feedforward neural
  networks.,''
\newblock in {\em {AISTATS}}, 2010, vol.~9, pp. 249--256.

\bibitem{liao13admm}
W.-C. Liao, M.~Hong, Y.-F. Liu, and Z.-Q. Luo,
\newblock ``Base station activation and linear transceiver design for optimal
  resource management in heterogeneous networks,''
\newblock {\em IEEE Transactions on Signal Processing}, vol. 62, no. 15, pp.
  3939--3952, 2014.

\bibitem{karipidis2006crosstalk}
E.~Karipidis, N.~Sidiropoulos, A.~Leshem, L.~Youming, R.~Tarafi, and M.~Ouzzif,
\newblock ``Crosstalk models for short {VDSL2} lines from measured 30mhz
  data,''
\newblock {\em EURASIP Journal on applied signal processing}, vol. 2006, pp.
  90--90, 2006.

\bibitem{karipidis2005experimental}
E.~Karipidis, N.~Sidiropoulos, A.~Leshem, and L.~Youming,
\newblock ``Experimental evaluation of capacity statistics for short {VDSL}
  loops,''
\newblock {\em IEEE Transactions on Communications}, vol. 53, no. 7, pp.
  1119--1122, 2005.

\end{thebibliography}
}

\begin{appendices}
\section{Proofs}

\begin{proof}[Proof of Lemma \ref{lemmada}]
Let us denote the approximated ${x}/{y}$  as $\widetilde{(x/y)}$, with $(x,y)\in {S}$. The idea is to approximate ${x}/{y}$ by its finite binary expansison $\widetilde{(x/y)} = \sum_{i=-n}^m{2^i z_i}$, where $z_i$ is the $i$-th digit of binary expansion of ${x}/{y}$.  The approximation error will be bounded by

\begin{equation}
 \left|\frac{x}{y} -\widetilde{\frac{x}{y}}\right| \leq \frac{1}{2^n}.
\end{equation}
To proceed, let us define
$$ m := \lceil \log(V_{\max}) \rceil,$$
where `$\log$' is logarithm with base 2. The $i$-th digit of binary expansion of ${x}/{y}$ (the bit corresponding to $2^i$) can be written as

\begin{equation}
z_i = \mbox{Binary}[(z^{(i)}- 2^iy)\geq 0]
\end{equation}
where $z^{(i)}$ is defined below recursively 
{
\begin{align*}
z^{(i)} = z^{(i+1)} - \max(2^{i+1}y + M(z_{i+1}-1),0), 
\end{align*}
with $M \geq 2^m Y_{\max}$ being a large constant and $ z^{(m)} = x$. }
Note that $z^{(i)}$ can be interpreted as ${x}/{y}\ \mbox{mod}\ 2^{i+1}$.
\begin{figure}
	\centerline{\includegraphics[width=9cm]{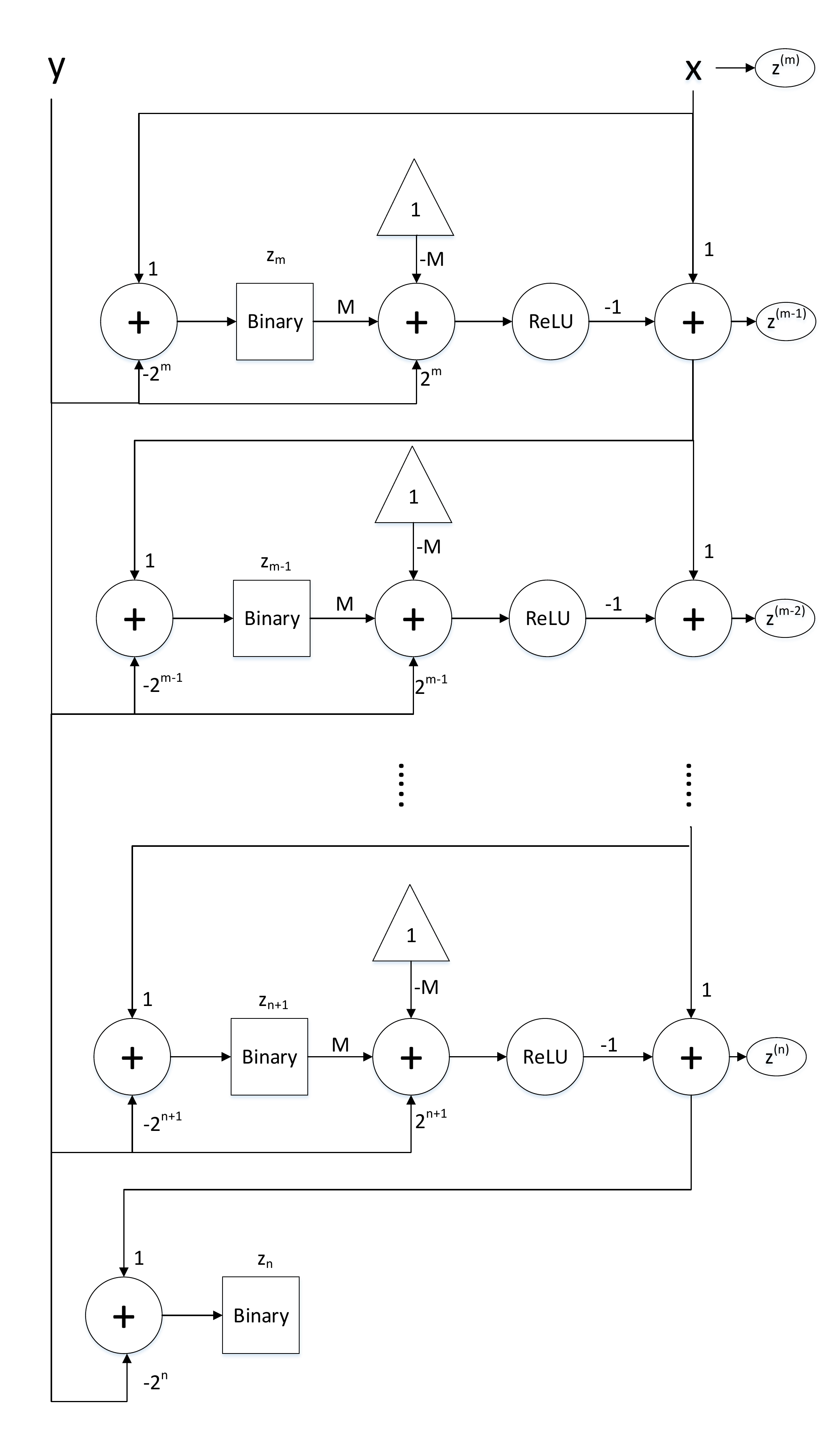}}
	\caption{\footnotesize Illustration of the network architecture for approximating $x/y$.}
	\label{fig:lemmada}
\end{figure}
{From the above construction, it is clear that we need two layers  to represent each {$z_i$}, where one ReLU in the first layer is used to represent $z^{(i)}$, and one Binary unit in second layer is used to represent $z_i$. {An illustration can be found in Fig. \ref{fig:lemmada}.}

Overall, we need $2(m+n+1)$ layers  and $(m+n+1)$ Binary units and $(m+n)$ ReLUs to represent all $z_i$'s. Once all $z_i$'s are constructed, $\widetilde{(x/y)}$ is given by}
\begin{equation*} 
\sum_{i=-n}^m{2^i z_i}. 
\end{equation*}
This completes the proof. 
\end{proof}

\begin{proof}[Proof of Lemma \ref{lemmama}]
The polynomial $xy$ can be approximated by $\tilde{x}y$ where $\tilde{x}$ is the $n$-bit binary expansion of $x$.
\begin{align}\label{multiply}
\begin{split}
\tilde{x}y =& \sum_{i=-n}^m 2^ix_i y 
= \sum_{i=-n}^m\max(2^iy + 2^{m+\lceil\log(Y_{\max})\rceil}(x_i-1),0)
\end{split}
\end{align}
where $x_i$ is $i$-th bit of binary expansion of $x$ and we assume $y \leq Y_{\max}, \forall y$.

Clearly, the approximation error will be at most
\begin{equation}\label{mulerror}
|xy - \tilde{x}y| =(x-\tilde{x})y \leq \frac{Y_{\max}}{2^n}
\end{equation}
since $(x-\tilde{x}) \leq \frac{1}{2^n}$.

The $i$-th digit of binary expansion of $x$ (the bit corresponding to $2^i$) can be done by substituting $y=1$ into $x/y$ in the proof of Lemma \ref{lemmada}, leading to
\begin{equation*}
x_i = \mbox{Binary}[x^{(i)} - 2^i \geq 0]
\end{equation*}
where $x^{(i)}$ have following recurrent definition
\begin{equation*}
x^{(i)} = \max(x^{(i+1)} -2^{i+1} ,0) 
\end{equation*}
and $x^{(m)} = x$, $x^{(i)}$ can be interpreted as $x\ mod\ 2^{i+1}$.

{Obviously, the number of layers and nodes for multiplication is similar to which in the proof of Lemma \ref{lemmada}, the only difference is that one more layer and $(m+n+1)$ ReLUs are added to achieve equation (\ref{multiply}).} 
\end{proof}

\begin{proof}[Proof of Lemma \ref{lemmaed}]
Given the following 
$$0 < Y_{\min} < y - \epsilon_2, 0 \leq  \epsilon_1 \leq x, 0 \leq \epsilon_2 \leq y , \frac{x}{y} \leq Z_{\max},$$ 
we have 
\begin{align*} 
\left|\frac{x}{y} - \frac{x \pm \epsilon_1}{y \pm \epsilon_2}\right|  
= \left|\frac{\pm x\epsilon_2 \pm y \epsilon_1}{y(y \pm \epsilon_2)}\right| \leq  \left(\frac{x}{y} \frac{\epsilon_2}{y-\epsilon_2}+\frac{\epsilon_1}{y-\epsilon_2}\right)  
\leq  \frac{Z_{\max} + 1}{Y_{\min}} \max(\epsilon_1,\epsilon_2).
\end{align*}
This completes the proof. 
\end{proof}

\begin{proof}[Proof of Lemma \ref{lemmaem}]
 Given 	$0 \leq x \leq X_{max},$
 $ 0 \leq y \leq Y_{max},$
 $\max(\epsilon_1,\epsilon_2) \leq \max(X_{\max},Y_{\max}),  \; \epsilon_1,\epsilon_2 \geq 0,$
 we have
\begin{align*} 
|xy - (x \pm \epsilon_1)(y \pm \epsilon_2)|  
\leq  x\epsilon_2 + y\epsilon_1 + \epsilon_1 \epsilon_2 
\leq  3\max(X_{\max},Y_{\max}) \max(\epsilon_1 ,  \epsilon_2)
\end{align*}
This completes the proof.  
\end{proof}

\begin{proof}[Proof of Theorem \ref{thm:app_wmmse}]
The idea is to approximate $v_k^t,a_k^t, b_k^t$ [defined by \eqref{ak} - \eqref{vk}], $k = {1,2,...,K}$ using a neural network with inputs $v_k^{t-1},a_k^{t-1}, b_k^{t-1}, k = {1,2,...,K}$ and $\{|h_{jk}|\}$. By repeating the network $T$ times and concatenating them together, we can approximate $T$ iterations of WMMSE.

Let us define  $\tilde{v}^t_k$ to be the output of the neural network that approximates the $t$th iteration of WMMSE. We will begin upper-bounding the error 
\begin{equation}
\epsilon_T \triangleq \max_{k}{|v_k^T - \widetilde{v}_k^T|}
\end{equation} 
and derive the relationship between upper bound of approximation error and network size.

To proceed, note that the  update rules of WMMSE algorithm can be rewritten in the following form  
\begin{align} 
a_k^t &:= u_k^t w_k^t|h_{kk}| = \frac{|h_{kk}|^2v_k^t}{\sum_{j=1,j \neq k}^K|h_{kj}|^2(v_j^t)^2 + \sigma_k^2}, \label{ak}\\
b_k^t &:= u_k^t (w_k^t)^2 \label{bk}\\
& =  \frac{|h_{kk}|^2(v_k^t)^2}{(\sum_{j=1,j \neq k}^K|h_{kj}|^2(v_j^t)^2 + \sigma_k^2) (\sum_{j=1}^K|h_{kj}|^2(v_j^t)^2 + \sigma_k^2)	} \nonumber\\
& =   \frac{a_k^t v_k^t}{\sum_{j=1}^K|h_{kj}|^2(v_j^t)^2 + \sigma_k^2} ,\nonumber\\
v_k^t &:= \bigg [\frac{\alpha_k a_k^{t-1}}{\sum_{j=1}^K \alpha_j b_j^{t-1}|h_{jk}|^2 } \bigg]_{0}^{\sqrt{P_{\max}}}. \label{vk}
\end{align}
We will approximate each iteration of WMMSE algorithm by a deep neural network and repeat the network $T$ times to approximate $T$ iterations of WMMSE. 
By using inequalities (\ref{diverrorprop}) and (\ref{mulerrorprop}), we can derive the upper bound of approximation error after $T$ iterations.

Firstly, we assume 
\begin{align} \label{epsilon}
\begin{split}
&\alpha_{\min} \leq \alpha_j \leq \alpha_{\max}, \; \forall j,  \\
&H_{\min} \leq |h_{kj}| \leq H_{\max}, \; \forall k,j, \\
&|\tilde{v_k}^{t} - v_k^{t}| \leq  \epsilon,\; {\forall~K,t.} 
\end{split}
\end{align}
Furthermore, we assume for any $h \in \mathcal{H}$ and $t = 0,1,2,...,T$, $\sum_{j=1}^K v_k^t \geq P_{\min}$.  

We note that $a_k^t$ is a rational function of $h_{kj}$ and $v_j^t, j = 1,2,...,K$, and we can approximate any rational function by repeatedly applying approximation to multiplication and division operations. We can construct an approximated version of $a_k^t$  by the following
\begin{equation} \label{a_approx}
\tilde{a}_k^t =   \frac{\widetilde{(|h_{kk}|^2\tilde{v_k}^t)}}{\sum_{j=1,j \neq k}^K\widetilde{|h_{kj}|^2\widetilde{(\tilde{v_j}^t \tilde{v_j}^t)}}+ \sigma_k^2} 
\end{equation} 
where $\widetilde{xy}$ means we will approximate $xy$ using a neural network by leveraging Lemma \ref{lemmama}.  
{From (\ref{epsilon}) and Lemma \ref{lemmama}, the error in numerator of (\ref{a_approx}) is upper-bounded by} 
\begin{align*}
\left|\widetilde{|h_{kk}|^2\tilde{v}_k^t} -  |h_{kk}|^2v_k^t\right| 
= & |h_{kk}|^2\left| \tilde{v}_k^t - \zeta - v_k^t \right| \\
\leq & |h_{kk}|^2\left|{v_k}^t - \epsilon -\frac{1}{2^n} -v_k^t \right| \\
\leq & H_{\max}^2 (\epsilon + \frac{1}{2^n}) \\ 
\leq & 2H_{\max}^2 \max\left(\epsilon,\frac{1}{2^n}\right),
\end{align*}
where the first equality is from $n$-bit binary expansion of $\tilde{v}_k^t$ in $\widetilde{|h_{kk}|^2\tilde{v}_k^t}$ and $\zeta < \frac{1}{2^n}$ is error induced by this binary expansion.

Similarly, for denominator of (\ref{a_approx}), we have  
\begin{align*}
&|\widetilde{(\tilde{v_j}^t \tilde{v_j}^t)} - ({v_j}^t)^2| \\
= & |(\tilde{v_j}^t - \zeta) \tilde{v_j}^t - ({v_j}^t)^2| \\
\leq & \max\bigg( |({v_j}^t -\epsilon - \frac{1}{2^n}) ({v_j}^t - \epsilon) -({v_j}^t)^2 |, \\
&|({v_j}^t +\epsilon) ({v_j}^t + \epsilon) -({v_j}^t)^2 |\bigg) \\
\leq & 3 \max({v_j}^t,\epsilon , 2^{-n}) \max(\epsilon , 2^{-n}) \\
\leq & 3\sqrt{P_{\max}} \max(\epsilon , 2^{-n})
\end{align*}
and 
\begin{align*}
&|\widetilde{|h_{kj}|^2\widetilde{(\tilde{v_j}^t \tilde{v_j}^t)}} -|h_{kj}|^2 ({v_j}^t)^2|\\
= & |h_{kj}|^2|\widetilde{\tilde{v_j}^t \tilde{v_j}^t} - \zeta - ({v_j}^t)^2| \\
\leq & |h_{kj}|^2 (3\sqrt{P_{\max}}\max(\epsilon , 2^{-n}) + 2^{-n}) \\
\leq & 4|h_{kj}|^2 \sqrt{P_{\max}}\max(\epsilon , 2^{-n}) \\
\leq & 4H_{\max}^2 \sqrt{P_{\max}}\max(\epsilon , 2^{-n}).
\end{align*}
Then, by Lemma \ref{lemmaed} and Lemma \ref{lemmada}, the approximation error of (\ref{a_approx}) is upper bounded by 
\begin{align*}
|\tilde{a_k}^t -  a^t_k| \leq& \bigg(4\frac{\sigma_k^2+H_{\max}^2\sqrt{P_{\max}}}{\sigma_k^4}(K-1)H_{\max}^2\sqrt{P_{\max}}+1\bigg) \\
 &\times\max(\epsilon,2^{-n})
\end{align*}
Using the same method, we can approximate $b_k^t$ by
\begin{equation*}
\tilde{b}_k^t = \frac{\widetilde{\tilde{a_k}^t \tilde{v_k}^t}}{\sum_{j=1}^K\widetilde{|h_{kj}|^2(\widetilde{\tilde{v_j}^t\tilde{v_j}^t})} + \sigma_k^2}.
\end{equation*}
And the error is upper bounded by 
\begin{align}
\begin{split}
\left|\tilde{b_k}^t -  b^t_k\right| \leq & \bigg(12\frac{\sigma_k^2+H^2_{\max}\sqrt{P_{\max}}}{\sigma_k^4}\frac{\sigma^4_k+H_{\max}^2P_{\max}}{\sigma_k^8} \\
&(K-1)H^2_{\max}\sqrt{P_{\max}}+1\bigg)\max(\epsilon,2^{-n}).
\end{split}
\end{align}
{The last result is obtained by using Lemma \ref{lemmama} - \ref{lemmaem}.}

Finally, we approximate $v_k^{t+1}$ as 
\begin{equation} \label{v_approx}
\tilde{v}_k^{t+1} =  \left[\frac{\alpha_k a_k^{t}}{\sum_{j=1}^K \alpha_j \widetilde{b_j^{t}|h_{jk}|^2} } \right]_{0}^{\sqrt{P_{\max}}}. 
\end{equation}
By using  Lemma \ref{lemmama} - \ref{lemmaem} again, the approximation error of (\ref{v_approx}) can be upper-bounded by the following 
\begin{align} \label{error_prop}
\begin{split}
|\tilde{v}_k^{t+1} - v_k^{t+1}| 
\leq& \bigg(\frac{1 }{(K\alpha_{\min})^2\sigma_l^2 H_{\min}^8P_{\min}^2}  \bigg(K \alpha_{\min}\sigma_l^2 H_{\min}^4P_{\min} \\
&+ \alpha_k H_{\max}^2\sqrt{P_{\max}} ((K-1) H_{\max}^2P_{\max} + \sigma_k^2) \\
&\times (KH_{\max}^2P_{\max} + \sigma_k^2) \bigg) \\
& \times \sum_{j=1}^K \alpha_j ((K-1) H_{\max}^2P_{\max} + \sigma_k^2)(KH_{\max}^2P_{\max} + \sigma_k^2) \\
& \times H_{\max}^2\bigg(12\frac{\sigma_k^2+H^2_{\max}\sqrt{P_{\max}}}{\sigma_k^4}\frac{\sigma^4_k+H_{\max}^2P_{\max}}{\sigma_k^8}\\
& \times (K-1)H^2_{\max}\sqrt{P_{\max}}+1\bigg) +1 \bigg)\max(\epsilon,2^{-n})\\
:=& G \times  \max(\epsilon,2^{-n})
\end{split}
\end{align}
where the constant $G$ is given below
\begin{align} 
\begin{split} \label{erroramplifier}
G \triangleq &  \frac{1 }{(K\alpha_{\min})^2\sigma_l^2 H_{\min}^8P_{\min}^2}  \bigg(K \alpha_{\min}\sigma_l^2 H_{\min}^4P_{\min} \\
&+ \alpha_k H_{\max}^2\sqrt{P_{\max}} ((K-1) H_{\max}^2P_{\max} + \sigma_k^2) \\
&\times (KH_{\max}^2P_{\max} + \sigma_k^2) \bigg) \\
& \times \sum_{j=1}^K \alpha_j ((K-1) H_{\max}^2P_{\max} + \sigma_k^2)(KH_{\max}^2P_{\max} + \sigma_k^2) \\
& \times H_{\max}^2\bigg(12\frac{\sigma_k^2+H^2_{\max}\sqrt{P_{\max}}}{\sigma_k^4}\frac{\sigma^4_k+H_{\max}^2P_{\max}}{\sigma_k^8}\\
& \times (K-1)H^2_{\max}\sqrt{P_{\max}}+1\bigg) + 1.
\end{split}
\end{align}

{Since we use the big $O$ notation in the theorem, in derivation of $G$, we assumed $H_{\max} \geq 1$ and $P_{\max} \geq 1$ so that the these two number will dominate constant 1 when propagating error, i.e. $H_{\max} \theta + \theta \leq 2H_{\max} \theta$ for any positive number $\theta$. 

Suppose we use $n$-bit binary expansion for all variables and  $\epsilon \geq 2^{-n}$, then \eqref{error_prop} reduces to 
\begin{align*}
&|\tilde{v}_k^{t+1} - v_k^{t+1}| \leq G \times  \epsilon \leq G \times |\tilde{v}_k^{t} - v_k^{t}|
\end{align*}
where the last inequality is due to (\ref{epsilon}). Thus, the approximation error is amplified at most $G$ times through approximation of one iteration.

Since $G \geq 1$, if we first approximate the initialization $v_k^0$ with $n$-bit binary expansion, the approximation error upper bound of latter iterations will be grater than $2^{-n}$, then using the method above to approximate $T$ iterations of WMMSE,} the resulting error will be upper bounded by
\begin{align}\label{finalerror}
2^{-n} {\underbrace{ G \times G ... \times G}_\text{$T$ consecutive multiplications of $G$}}   \leq \epsilon_T
\end{align}
which means 
\begin{equation}\label{n_min}
n \geq T \log G + \log \frac{1}{\epsilon_T}.
\end{equation}

By counting the number of neurons and layers used for approximation, we can derive the trade-off between the approximation error upper bound and the number of units (RelU and binary) and number of layers we use.

To approximate $a_k^t$, we need $\lceil\log(P_{\max})/2\rceil + n + 1$ bits for each $\tilde{v_k^t}$ (since $v_k^t \leq \sqrt{P_{max}}$), $\lceil\log(P_{\max})\rceil +n + 1$ bits for each $ \widetilde{\tilde{v_k}^t\tilde{v_k}^t}$, and $\lceil 2\log(H_{\max})+\log(P_{\max})/2 + 2\log{\frac{1}{\sigma_k}}\rceil+ n + 1$ bits for the division operation. Over all, we need at most
\begin{align}
(3K+1)\lceil\log(P_{\max})/2\rceil + \lceil 2\log(H_{\max})+ 2\log{\frac{1}{\sigma_k}}\rceil \nonumber
+ (2K+1)(n+1) 
\end{align}   
bits for each $a_k^t$.

For approximating $b_k^t$, we can use the binary expansion of $\tilde{v_k^t}$ and $\widetilde{\tilde{v_k}^t\tilde{v_k}^t}$ we obtained for $a_k^t$, thus we just need one more binary expansion process for the division operation. The required number of bits is 
\begin{equation}
\lceil 2\log (H_{\max}) + \log (P_{\max}) + 4\log(\frac{1}{\sigma_k})\rceil + n + 1.
\end{equation}

As for $v_k^{t+1}$, we already have binary expansion of $b_j^t$, then only the division operation needs to be approximated. Note that the projection operator $[ \cdot ]^{\sqrt{P_{\max}}}_{0}$ needed to generate $v_k^{t+1}$ can be implemented by two ReLU units:
\begin{equation}
O = \sqrt{P_{\max}} -  \max(\sqrt{P_{\max}} - \max(I,0), 0)
\end{equation}
where $I$ is the input and $O$ is projected value of $I$.  

The number of bits required for $v_k^{t+1}$ is 
\begin{equation}
\lceil \log(P_{\max})/2 \rceil + n + 1.
\end{equation}
Finally, the number of bits we used is at most
\begin{align}
\begin{split} \label{n_nodes}
TK(6\lceil\log(\frac{1}{\sigma_k})\rceil + 4 \lceil\log(H_{\max})\rceil  
+ (3K+4)(\lceil\log P_{\max}/2\rceil) + (2K+3)(n+1)).
\end{split}
\end{align}
where $n$ satisfies eq. \eqref{n_min}.

The number of layers is $1/K^2$ times of the number of bits due to parallelism, each bit can be implemented by a constant number of binary units and ReLUs, taking $O$ notation of (\ref{n_nodes}), we get Theorem \ref{thm:app_wmmse} .

\end{proof}

\end{appendices}
\end{document}